\documentclass{llncs}
\usepackage{amsmath,amssymb}
\usepackage[utf8]{inputenc}
\usepackage{todonotes}

\usepackage{thmtools}
\usepackage{chngcntr}

\def\N{\mathbb{N}}

\def\ed{wzb}
\def\ta{\mathtt{a}}
\def\tb{\mathtt{b}}

\DeclareMathOperator{\ScatFact}{ScatFact}

\def\nth#1{#1$^{\text{th}}$}

\title{$k$-Spectra of Weakly-$c$-Balanced Words}
\author{Joel D. Day\inst{1} \and Pamela Fleischmann\inst{2} \and Florin 
Manea\inst{2} \and Dirk Nowotka\inst{2} }
\authorrunning{J.D. Day \and P. Fleischmann \and F. Manea \and D. Nowotka}

\institute{Loughborough University, UK, \email{J.Day@lboro.ac.uk} \and 
Kiel University, 
Germany \email{\{fpa,flm,dn\}@informatik.uni-kiel.de}}

\begin{document}

\maketitle
\begin{abstract}
A word $u$ is a scattered factor of $w$ if $u$ can be obtained from $w$ by deleting some of its
letters. That is, there exist the (potentially empty) words $u_1,u_2,..., u_n$, and  
$v_0,v_1,..,v_n$ such that $u = u_1u_2...u_n$ and $w = v_0u_1v_1u_2v_2...u_nv_n$. We consider the 
set of length-$k$ scattered factors of a given word $w$,  called here $k$-spectrum and denoted 
$\ScatFact_k(w)$. We prove a series of properties of the sets $\ScatFact_k(w)$ for binary 
weakly-$0$-balanced and, respectively, weakly-$c$-balanced words $w$, i.e., words over a two-letter alphabet where the number of occurrences of each letter is 
the same, or, respectively, one letter has $c$-more occurrences than the other. 
In particular, we consider the question which cardinalities $n= 
|\ScatFact_k(w)|$ are obtainable, for a positive integer $k$, when $w$ is either a weakly-$0$-balanced binary word of 
length $2k$, or a weakly-$c$-balanced binary word of length $2k-c$. We also consider the problem of 
reconstructing words from their $k$-spectra.
\end{abstract}

\section{Introduction}

Given a word $w$, a scattered factor (also called scattered subword, or simply subword in the 
literature) is a word obtained by removing one or more factors from $w$. More formally, $u$ is a 
scattered factor of $w$ if there exist $u_1,\ldots,u_n\in\Sigma^{\ast}$, $v_0,  \ldots, 
v_n\in\Sigma^{\ast}$ such that $u=u_1u_2\ldots u_n$ and $w = v_0 u_1 v_1 u_1 \ldots u_n v_n$. 
Consequently a scattered factor of $w$ can be thought of as a representation of $w$ in which some parts 
are missing. As such, there is considerable interest in the relationship of a word and its 
scattered 
factors from both a theoretical and practical point of view. For an introduction to the study of scattered factors, 
see Chapter 6 of~\cite{Loth97}. 
On the one hand, it is easy to imagine how, in any situation where discrete, linear data is read 
from an imperfect input -- such as when sequencing DNA or during the transmission of a digital 
signal -- scattered factors form a natural model, as multiple parts of the input may be missed, but 
the rest will remain unaffected and in-sequence. For instance, various applications and connections of this model in verification are discussed in \cite{Zetzsche16,HalfonSZ17} within a language theoretic framework, while applications of the model in DNA sequencing are discussed in \cite{ElzingaRW08} in an algorithmic framework.
On the other hand, from a more algebraic  
perspective, there have been efforts to bridge the gap between the non-commutative field of 
combinatorics on words with traditional commutative mathematics via Parikh matrices (cf. 
e.g.,~\cite{Mat04,Salomaa05}) which are closely related to, and influenced by the topic of 
scattered 
factors.

The set (or also in some cases, multi-set) of scattered factors of a word $w$, denoted 
$\ScatFact(w)$ is typically exponentially large in the length of $w$, and contains a lot of 
redundant information in the sense that, for $k^\prime<k\leq|w|$, a word of length $k^\prime$ is a 
scattered factor of $w$ if and only if it is a scattered factor of a scattered factor of $w$ of 
length $k$. This has led to the idea of $k$-spectra: the set of all length-$k$ scattered factors of 
a word. For example, the $3$-spectrum of the word $\ta\tb\ta\tb\tb\tb$ is the set
$\{\ta\ta\tb, \ta\tb\ta, \ta\tb\tb, \tb\ta\tb, \tb\tb\tb\}$.
Note that unlike some literature, we do not consider the $k$-spectra to be the multi-set of 
scattered factors in the present work, but rather ignore the multiplicities. This distinction is 
non-trivial as there are significant variations on the properties based on these different 
definitions (cf. e.g.,~\cite{conf/dlt/Manuch99}).
Also, the notion of $k$-spectra is closely related to the classical notion of 
factor complexity of words, which counts, for each positive integer $k$, the 
number of distinct factors of length $k$ of a word. Here, the cardinality of the 
$k$-spectrum of a word gives the number of the word's distinct {\em scattered} 
factors of length $k$. 

One of the most fundamental questions about $k$-spectra of words, and indeed sets of scattered 
factors in general, is that of recognition: given a set $S$ of words (of length 
$k$), is $S$ the subset 
of a $k$-spectrum of some word? In general, it remains a long standing 
goal of the theory to give a ``nice'' descriptive characterisation of scattered factor sets (and 
$k$-spectra), and to better understand their structure~\cite{Loth97}.
Another fundamental question concerning $k$-spectra, and one well motivated in 
several 
applications, 
is the question of reconstruction: given a word $w$ of length $n$, what is the smallest value $k$ such that the 
$k$-spectrum of $w$ is uniquely determined? 
This question was addressed and solved successively in a variety of cases. In particular, in~\cite{dress04}, the exact bound of 
$\frac{n}{2}+1$ is given in the general case. Other variations, including for the definition of 
$k$-spectra where multiplicities are also taken into account, are considered 
in~\cite{conf/dlt/Manuch99}, while~\cite{Holub09} considers the question of reconstructing words 
from their palindromic scattered factors.

In the current work, we consider $k$-spectra in the restricted setting of 
a binary alphabet $\Sigma=\{\ta,\tb\}$. For such an alphabet, we can always identify the natural number $c\in\N_0$
which describes how weakly balanced a word is: $c$ is the difference between 
the amount of $\ta$s and $\tb$s. Thus, it seems natural to categorise all words 
over $\Sigma$ according to this difference: a binary word where one letter has 
exactly $c$ more occurrences than the other one is called weakly-$c$-balanced. In 
Section~\ref{min-max} the cardinalities of $k$-spectra of weakly-$c$-balanced words of length 
$2k-c$ are investigated. Our first results concern the minimal and maximal cardinality $\ScatFact_k$ might have. 
We show that the cardinality ranges for weakly-$0$-balanced between $k+1$ and $2^k$, and determine exactly for which words of length $2k$ these values are reached. 
In the case of weakly-$c$-balanced words, we are able to replicate the result regarding the minimal cardinality of $\ScatFact_k$, but the case of maximal cardinality seems to be more complicated. 
To this end, it seems that the words containing many alternations between the two letters of the alphabet have larger sets $\ScatFact_k$. Therefore, we first investigate the scattered factors of the words which are prefixes of 
$(\ta\tb)^\omega$ and give a precise description of all scattered factors of any length of such words. That is, not only we compute the cardinality of $\ScatFact_k(w)$, for all such words $w$, 
but also describe a way to obtain directly the respective scattered factors, without repetitions. We use this to describe exactly the sets $\ScatFact_i$ for the word $(\ta\tb)^{k-c}\ta^c$, which seems a good candidate for a weakly-$c$-balanced word with many distinct scattered factors. 

Further, in Section \ref{sec:card}, we explore more the cardinalities of 
$\ScatFact_k(w)$ for weakly-$0$-balanced words $w$ of length $2k$. We obtain 
for these words that the smallest three numbers which are possible 
cardinalities for their $k$-spectra are $k+1$, 
$2k$, and $3k-3$, thus identifying two gaps in the set of such cardinalities. 
Among other results on this topic, we show that for every constant $i$ there 
exist a word $w$ of length $2k$ such that $|\ScatFact_k(w)|\in \Theta(n^i)$; we 
also show how such a word can be constructed.

Finally, in Section \ref{reconstruct}, we also approach the question of reconstructing 
weakly-$0$-balanced words from $k$-spectra in the specific case that the spectra are also limited to 
weakly-$0$-balanced 
words only. While we are not able to resolve the question completely, we conjecture that the 
situation is similar to the general case: the smallest value $k$ such that the 
$k$-spectrum of $w$ is uniquely determined is $k = \frac{|w|}{2}+1$ if 
$\frac{|w|}{2}$ 
is odd and $k= \frac{|w|}{2} +2$, otherwise, in the case when $w$ contains at most two blocks of $\tb$s. 

After introducing a series of basic definitions, preliminaries, and 
notations, the organisation of the paper follows the description above. The proofs can be found in \cite{dlt2019}.

\section{Preliminaries}

Let $\N$ be the set of natural numbers, $\N_0 = \N\cup\{0\}$, and let 
$\N_{\geq k}$ be all natural numbers greater than or equal to $k$. 
Let $[n]$ denote the set $\{1,\ldots, n\}$ and $[n]_0 = [n]\cup\{0\}$
for an $n\in\N$.

We consider words $w$ over the alphabet $\Sigma=\{\ta,\tb\}$. 
$\Sigma^*$ denotes the set of all finite words over $\Sigma$, also called binary words. 
$\Sigma^{\omega}$ the set of all infinite words over $\Sigma$, also called binary infinite words.
The {\em empty word} is denoted by $\varepsilon$ and 
$\Sigma^+$ is the free semigroup $\Sigma^*\backslash\{\varepsilon\}$.
The length of a word $w$ is denoted by $|w|$. Let 
$\Sigma^{\leq k}:=\{w\in\Sigma^{\ast}|\,|w|\leq k\}$ and $\Sigma^k$ be the 
set of all words of length exactly $k\in\N$. The number of occurrences of a 
letter $\ta\in\Sigma$ in a word $w\in\Sigma^{\ast}$ is denoted by $|w|_\ta$.
The \nth{i} letter of a word $w$ is given by $w[i]$ for $i\in[|w|]$. For a given 
word $w\in\Sigma^n$ the {\em reversal} of $w$ is defined by 
$w^R=w[n]w[n-1]\dots w[2]w[1]$. The powers of $w\in\Sigma^{\ast}$ are 
defined recursively by $w^0=\varepsilon$, $w^n=ww^{n-1}$ for $n\in\N$.

A word $w\in\Sigma^{\ast}$ is called {\em weakly-$c$-balanced} if $||w|_a-|w|_b|=c$ 
for $c\in\N_0$. 
Thus weakly-$0$-balanced words have the same number of $\ta$s and $\tb$s. Let 
$\Sigma_{\ed}^{\ast}$ be the set of all weakly-$0$-balanced words over $\Sigma$. 
For example, $\ta\tb\ta\ta$ is weakly-$2$-balanced, $\ta\tb\ta$ is weakly-$1$-balanced, 
while $\ta\tb\tb\ta\tb\ta$ is weakly-$0$-balanced.

A word $u\in\Sigma^{\ast}$ is a \emph{factor} of $w\in\Sigma^{\ast}$, if 
$w=xuy$ holds for some words $x,y\in\Sigma^{\ast}$. Moreover, $u$ is a 
\emph{prefix} of $w$ if $x=\varepsilon$ holds and a \emph{suffix} if 
$y=\varepsilon$ holds. The factor of $w$ from the \nth{i} to the \nth{j} 
letter will be denoted by $w[i..j]$ for $0\leq i\leq j\leq |w|$. 
Given a letter $\ta\in\Sigma$ and a word $w\in\Sigma^{\ast}$, a \emph{block} of 
$\ta$ is a factor $u = w[i..j]$ with $u=\ta^{j-i}$, such that  either $i=1$ or 
$w[i-1]=\tb\neq \ta$ and either $j=|w|$ or $w[j+1]=\tb\neq \ta$. For example the 
word $\ta\tb\ta\ta\ta\tb\ta\ta\tb\tb$ has 3 $\ta$-blocks and 3 $\tb$-blocks.
Scattered factors and $k$-spectra are defined as follows.

\begin{definition}
A word $u=a_1\dots a_n\in\Sigma^n$, for $n\in\N$, is a {\em scattered factor} of 
a word 
$w\in\Sigma^+$ if there exists $v_0,\dots,v_n\in\Sigma^{\ast}$ with 
$w=v_0a_1v_1\dots 
v_{n-1}a_nv_n$. Let $\ScatFact(w)$ denote the set of $w$'s scattered factors and 
consider 
additionally $\ScatFact_k(w)$ and $\ScatFact_{\leq k}(w)$ as the two subsets of 
$\ScatFact(w)$ 
which contain only the scattered factors of length $k\in\N$ or the ones up to 
length $k\in\N$.
\end{definition}

The sets $\ScatFact_{\leq k}(w)$ and $\ScatFact_k(w)$ are also known as {\em 
full 
$k$-spectrum} and, respectively, 
{\em $k$-spectrum} of a word $w\in\Sigma^{\ast}$ (see \cite{BerKar03}, 
\cite{conf/dlt/Manuch99}, 
\cite{RozSal97}) and moreover, scattered factors are often called {\em subwords} 
or {\em 
scattered subwords}. Obviously the $k$-spectrum is empty for $k>|w|$ and 
contains 
exactly $w$'s letters for $k=1$ and only $w$ for $k=|w|$. Considering the word 
$w=\ta\tb\tb\ta$, 
the other spectra are given by $\ScatFact_2(w)=\{\ta^2,\tb^2,\ta\tb,\tb\ta\}$ 
and 
$\ScatFact_3(w)=\{\ta\tb^2,\ta\tb\ta,\tb^2\ta\}$.

It is worth noting that if $u$ is a scattered factor of $w$, and $v$ is a 
scattered factor of $u$, then $v$ is a scattered factor of $w$. 
Additionally, notice two  important symmetries regarding $k$-spectra. 
For $w \in 
\Sigma^*$ and the {\em renaming morphism}
$\overline{\cdot}:\Sigma\rightarrow\Sigma$ with $\overline{\ta}=\tb$ and 
$\overline{\tb}=\ta$ we have $
\ScatFact(w^R) = \{u^R \mid u \in \ScatFact(w)\}$ and $\ScatFact(\overline{w}) 
= 
\{ \overline{u} \mid u \in \ScatFact(w)\}$.
Thus, from a structural point of 
view, it is sufficient 
to consider only one 
representative from the equivalence classes induced by the equivalence relation 
where $w_1$ is equivalent to $w_2$ whenever 
$w_2$ is obtained by a composition of reversals and renamings from $w_1$. 
Considering w.l.o.g. the order $\ta<\tb$ on $\Sigma$, we 
choose the 
lexicographically smallest word as representative from each class. As such, we will mostly analyse the $k$-spectra of 
words starting with $\ta$. We shall 
make 
use of this fact 
extensively in Section~\ref{sec:card}.

\section{Cardinalities of $k$-Spectra of Weakly-$c$-Balanced Words}\label{sec:card}

\label{min-max}
In the current section, we consider the combinatorial properties of $k$-spectra of 
weakly-$c$-balanced finite words. In particular, we are interested in the cardinalities of the 
$k$-spectra and in the question: which cardinalities are (not) possible? Since 
the $k$-spectra of $\ta^n$ and $\tb^n$ are just $\ta^k$ and $\tb^k$ 
respectively for all $n\in\N_0$ and $k\in[n]_0$, we assume 
$|w|_{\ta},|w|_{\tb}>0$ for $w\in\Sigma^{\ast}$.
It is 
a straightforward observation that not every subset of $\Sigma^k$ is a 
$k$-spectrum of some word $w$. For example, for $k=2$, $\ta\ta$ and $\tb\tb$ can 
only be scattered factors of a word containing both $\ta$s and $\tb$s, and 
therefore having either $\ta\tb$ or $\tb\ta$ as a scattered factor as well. Thus, there is no word $w$ such that $\ScatFact_2(w)=\{\ta\ta,\tb\tb\}$.

In general, for any word containing only $\ta$s or only $\tb$s, there will be exactly one scattered factor of each length, while for words containing both $\ta$'s and $\tb$'s, the smallest $k$-spectra are realised for words of the form $w= \ta^n\tb$ (up to renaming and reversal), for which $\ScatFact_k(w) = \{\ta^k, \ta^{k-1}\tb\}$ for each $k \in [|w|]$. On the other hand, as Proposition~\ref{lemfull} shows, the maximal $k$-spectra are those containing all words of length $k$ -- and hence have size $2^k$, achieved by e.g. $ w= (\ta\tb)^n$ for $n \geq k$. Note that when weakly-$0$-balanced words are considered, the same maximum applies, since $(\ta\tb)^n$ is weakly-$0$-balanced, while the minimum does not, since $\ta^n\tb$ is not weakly-$0$-balanced. 

It is straightforward to enumerate all possible $k$-spectra, and describe the words realising them for $k \leq 2$, hence we shall generally consider only $k$-spectra in the sequel for which $k \geq 3$.
Our first result generalises the previous 
observation about minimal-size $k$-spectra.

\setcounter{theorem}{\value{definition}}
\begin{theorem}\label{lemsmallest}
For $k \in \mathbb{N}_{\geq 3}$, $c\in[k-1]_0$, $i \in [c]_0$, and a weakly-$c$-balanced word 
$w\in\Sigma^{2k-c}$, we have 
$|\ScatFact_{k-i}(w)|\geq k-c+1$, where equality holds if and only if 
$w\in\{\ta^k\tb^{k-c},\ta^{k-c}\tb^k,\tb^k\ta^{k-c},\tb^{k-c}\ta^k\}.$
Moreover, if $w\in\Sigma_{\ed}^{2k}\backslash\{\ta^k\tb^k\}$, then $|\ScatFact_{k}(w)|\geq 
k+3$.
\end{theorem}
\refstepcounter{definition}

\begin{proof}
Consider firstly only weakly-$0$-balanced words, i.e. $c=0$ and w.l.o.g. only 
$w=\ta^k\tb^k$. The cases $k=1$ and $k=2$ are the induction basis.

The word $\ta^k\tb^k$ has obviously all $\ta^r\tb^s$ for $r,s\in[k]_0$ as 
scattered factors, thus 
$k+1$ many. This proves the $\Leftarrow$-direction.

Consider now a word $w\in\Sigma_{\ed}^{2k}\backslash\{\ta^k\tb^k,\tb^k\ta^k\}$.  
 Since 
$w$ is not $\ta^k\tb^k$, $w$ contains a factor $\ta\tb\ta$ or 
$\tb\ta\tb$. Assume w.l.o.g. that $w=x\ta\tb\ta y$ holds for 
$x,y\in\Sigma^{\ast}$ with 
$|x|+|y|=2k-3$. By $w\in\Sigma_{\ed}^{2k}$ follows that $|x|_{\tb}$ or 
$|y|_{\tb}$ is not zero. Choose 
w.l.o.g. $z_1,z_2\in\Sigma^{\ast}$ with $y=z_1\tb z_2$ which implies 
$w=x\ta\tb\ta z_1\tb z_2$. Consequently 
$|xz_1z_2|_a=|xz_1z_2|_b=k-2$ holds.\\
{\tt case 1:} $xz_1z_2=\ta^{k-2}\tb^{k-2}$\\
By induction $|\ScatFact_{k-2}(xz_1z_2)|=(k-2)+1=k-1$. Let $u$ be a scattered 
factor of 
$xz_1z_2$ of length $k-2$. Then there 
exist $u_1,u_2$, and $u_3$ such that $u_1$ is a scattered factor of $x$, $u_2$ 
of $z_1$, and 
$u_3$ of $z_3$ respectively. Consequently
\[
u_1\ta\ta u_2u_3,\quad u_1\ta\tb u_2u_3,\quad\mbox{and}\quad u_1\tb\ta u_2u_3
\]
are different elements of $\ScatFact_k(w)$. 
Each scattered factor of $xz_1z_2$ is of the form $\ta^r\tb^s$ for 
$r,s\in[k-2]_0$. We will now prove in which cases the aforementioned scattered 
factors are different. Consider $u=u_1u_2u_3=\ta^r\tb^s$ and 
$u'=u_1'u_2'u_3'=\ta^{r'}\tb^{s'}$ to be different scattered factors of this 
form, i.e. $r\neq r'$ and $s\neq s'$. Set 
\begin{align*}
\alpha_1 &= u_1\ta\ta u_2u_3, & \beta_1 &= u_1'\ta\ta u_2'u_3'\\
\alpha_2 &= u_1\tb\ta u_2u_3, & \beta_2 &= u_1'\tb\ta u_2'u_3'\\
\alpha_3 &= u_1\ta\tb u_2u_3, & \beta_3 & = u_1\ta\tb u_2u_3.
\end{align*}
If $u_1=\ta^{r_1}$, $u_2u_3=\ta^{r_2}\tb^s$ and $u_1'=\ta^{r_1'}$, 
$u_2'u_3'=\ta^{r_2'}\tb^{s'}$ with $r_1+r_2=r$ and $r_1'+r_2'=r'$,  we get
because of $r\neq r'$, $r_1\neq -1$, 
\begin{align*}
\alpha_1=\ta^{r+2}\tb^s&\neq \ta^{r'+2}\tb^s=\beta_1,\\
\alpha_1=\ta^{r+2}\tb^s&\neq \ta^{r_1'}\tb\ta^{r_2'+1}\tb^{s'}=\beta_2\\
\alpha_2=\ta^{r_1}\tb\ta^{r_2+1}\tb^s&\neq 
\ta^{r_1'}\tb\ta^{r_2'+1}\tb^{s'}=\beta_2.
\end{align*}
If $u_1=\ta^{r_1}$, $u_2u_3=\ta^{r_2}\tb^s$ and $u_1'=\ta^{r'}\tb^{s_1'}$, 
$u_2'u_3'=\tb^{s_2'}$ with $r_1+r_2=r$, $s_1'+s_2'=s'$, and $s_1'\neq 0$ 
(already in the previous case) we get
because of $s_1'\neq 0$, 
\begin{align*}
\alpha_1=\ta^{r+2}\tb^s&\neq \ta^{r'}\tb^{s_1'}\ta\ta\tb^{s_2'}=\beta_1,\\
\alpha_1=\ta^{r+2}\tb^s&\neq \ta^{r'}\tb^{s_1'}\tb\ta\tb^{s_2'}=\beta_2\\
\alpha_2=\ta^{r_1}\tb\ta^{r_2+1}\tb^s&\neq 
\ta^{r'}\tb^{s_1'}\tb\ta\tb^{s_2'}=\beta_2.
\end{align*}
If $u_1=\ta^{r}\tb^{s_1}$, $u_2u_3=\tb^{s_2}$ and $u_1'=\ta^{r'}\tb^{s_1'}$, 
$u_2'u_3'=\tb^{s_2'}$ with $r_1+r_2=r$, $s_1'+s_2'=s'$, and $s_1,s_1'\neq 0$ 
(already in the previous case) we get
because of $r'\neq r$ and $s_1,s_1'\neq 0$, 
\begin{align*}
\alpha_1=\ta^{r}\tb^{s_1}\ta\ta\tb^{s_2}&\neq 
\ta^{r'}\tb^{s_1'}\ta\ta\tb^{s_2'}=\beta_1,\\
\alpha_1=\ta^{r}\tb^{s_1}\ta\ta\tb^{s_2}&\neq 
\ta^{r'}\tb^{s_1'}\tb\ta\tb^{s_2'}=\beta_2\\
\alpha_2=\ta^{r}\tb^{s_1}\tb\ta\tb^{s_2}&\neq 
\ta^{r'}\tb^{s_1'}\tb\ta\tb^{s_2'}=\beta_2.
\end{align*}
Consequently $\alpha_1$ and $\alpha_2$ are all different and we get $2(k-1)$ 
many different scattered factors. Assume now additionally $|r-r'|=3$.
If $u_1=\ta^{r_1}$, $u_2u_3=\ta^{r_2}\tb^s$ and $u_1'=\ta^{r_1'}$, 
$u_2'u_3'=\ta^{r_2'}\tb^{s'}$ with $r_1+r_2=r$ and $r_1'+r_2'=r'$,  we get
because of $s_1'\neq 0$, $r'\neq r$, $r'\neq r+1$
\begin{align*}
\alpha_1=\ta^{r+2}\tb^s&\neq = \ta^{r_1'}\ta\tb\ta^{r_2'}\tb^{s'}=\beta_3,\\
\alpha_2=\ta^{r_1}\tb\ta^{r_2+1}\tb^s&\neq 
\ta^{r_1'}\ta\tb\ta^{r_2'}\tb^{s'}=\beta_3,\\
\alpha_3=\ta^{r_1}\ta\tb\ta^{r_2}\tb^s&\neq 
\ta^{r_1'}\ta\tb\ta^{r_2'}\tb^{s'}=\beta_3,\\
\end{align*}
If $u_1=\ta^{r_1}$, $u_2u_3=\ta^{r_2}\tb^s$ and $u_1'=\ta^{r'}\tb^{s_1'}$, 
$u_2'u_3'=\tb^{s_2'}$ with $r_1+r_2=r$, $s_1'+s_2'=s'$, and $s_1'\neq 0$ 
(already in the previous case) we get
because of $s_1'\neq 0$, $r'\neq r+2$,
\begin{align*}
\alpha_1=\ta^{r+2}\tb^s&\neq \ta^{r'}\tb^{s_1'}\ta\tb\tb^{s_2'}=\beta_3,\\
\alpha_2=\ta^{r_1}\tb\ta^{r_2+1}\tb^s&\neq 
\ta^{r'}\tb^{s_1'}\ta\tb\tb^{s_2'}=\beta_3\\
\alpha_3=\ta^{r_1}\ta\tb\ta^{r_2}\tb^s&\neq 
\ta^{r'}\tb^{s_1'}\ta\tb\tb^{s_2'}=\beta_3.
\end{align*}
If $u_1=\ta^{r}\tb^{s_1}$, $u_2u_3=\tb^{s_2}$ and $u_1'=\ta^{r'}\tb^{s_1'}$, 
$u_2'u_3'=\tb^{s_2'}$ with $r_1+r_2=r$, $s_1'+s_2'=s'$, and $s_1,s_1'\neq 0$ 
(already in the previous case) we get
because of $r'\neq r$ and $s_1,s_1'\neq 0$, $r'\neq r+2$,
\begin{align*}
\alpha_1=\ta^{r}\tb^{s_1}\ta\ta\tb^{s_2}&\neq 
\ta^{r'}\tb^{s_1'}\ta\tb\tb^{s_2'}=\beta_3,\\
\alpha_2=\ta^{r}\tb^{s_1}\tb\ta\tb^{s_2}&\neq 
\ta^{r'}\tb^{s_1'}\ta\tb\tb^{s_2'}=\beta_3\\
\alpha_3=\ta^{r}\tb^{s_1}\ta\tb\tb^{s_2}&\neq 
\ta^{r'}\tb^{s_1'}\ta\tb\tb^{s_2'}=\beta_3.
\end{align*}
Consequently we have another $\lfloor\frac{k-2}{3}\rfloor+1$ different 
scattered factors. This sums up to $|\ScatFact_k(w)|\geq 
\frac{7k-8}{3}>k+1$. An immediate result is that the $k$-spectrum has 
at least $k+3$ elements for $k\geq 5$. For $k=3$ and $k=4$ the results 
can be easily verified by testing.\\
{\tt case 2:} $xz_1z_2\neq \ta^{k-2}\tb^{k-2}$\\
In this case all words of the form $\ta^r\ta\tb\ta\ta^s$ for $r+s=k-3$, 
$r\in[|x|_\ta]_0$, and 
$s\in[|y|_\ta]_0$ are $|x|_{\ta}+1$ different scattered factors of length $k$ 
of $w$. 
Analogously all 
$\tb^{r'}\ta\tb\ta\tb^{s'}$ with $r'+s'=k-3$, $r'\in[|x|_\tb]_0$, 
$s'\in[|y|_\tb]_0$ are $|x|_{\tb}+1$ 
different scattered factors of length $k$ of $w$. All these factors are 
different and additionally 
$w$ has $\ta^k$ and $\tb^k$ as scattered factors. Hence $|\ScatFact_k(w)|\geq 
|x|_{\ta}+|x|_{\tb}+4=|x|+4$ holds. Since the 
length of $w$ is $2k$, the length of $xy$ is $2k-3$ and consequently $x$ and $y$ 
have different 
lengths. Assume w.l.o.g. $|x|>|y|$, i.e. $|x|\geq k-1$. This implies 
$|\ScatFact_k(w)|\geq k+3$ follows. This proves the claim for $c=0$.

Assume now $c>0$ and let $w=\ta^k\tb^{k-c}$. By the previous part we know 
$|\ScatFact_{k-c}(w)|=k-c+1$ if and only if $w=\ta^{k-c}\tb^{k-c}$. The claim 
about the 
$(k-c)$-spectrum follows immediately by 
$\ScatFact_{k-c}(w)=\ScatFact_{k-c}(\ta^k\tb^{k-c})$ since the prepended $\ta$s 
do not change the $(k-c)$-spectrum. For $i\in[c-1]_0$ notice 
that $x\in\ScatFact_{k-i}(\ta^k\tb^{k-c})$ implies that $\ta x$ (resp. 
$x\tb$, $x\ta$, $\tb x$) is a scattered factor of $\ta^k\tb^{k-c}$ of length 
$k-i+1$. Thus 
$|\ScatFact_{k-i+1}(w)|\geq k-c+1$ follows. 
On the other hand a scattered factor of $\ta^k\tb^{k-c}$ of length $k-i+1$ 
is exactly of this form, since it can neither start with $\tb$ 
($\ta^k\tb^{k-c}$ has only 
$(k-c)$ occurrences of $\tb$) nor 
contain $\tb\ta$ resp. $\ta\tb$ (this would be the implication of a scattered 
factor being of the 
form $\ta x'$ with $|x'|=k-i$, $x'\not\in\ScatFact_{k-i}(\ta^k\tb^{k-c})$). \qed
\end{proof}

\setcounter{remark}{\value{definition}}
\begin{remark}\label{smallestgap}
Theorem~\ref{lemsmallest} 
 answers immediately the question, whether a given set 
$S\subseteq\Sigma^{k}$, with $|S| < k+1$ or $|S| = k+2$, is a $k$-spectrum of a word $w\in\Sigma_{\ed}^{2k}$ in the negative. \end{remark}
\refstepcounter{definition}

Theorem~\ref{lemsmallest} shows  that the smallest cardinality of the $k$-spectrum of a word $w$
is reached when the letters in $w$ are {\em nicely ordered}, both for 
weakly-$0$-balanced words as well as for weakly-$c$-balanced words with $c>0$. The 
largest cardinality is, not surprisingly, reached for words where the alternation of $\ta$ and $\tb$ letters is, in a sense, maximal, e.g., for $w=(\ta\tb)^k$. To this end, one can show a general result.

\setcounter{theorem}{\value{definition}}
\begin{theorem}\label{arbitraryw}
For $w\in\Sigma^{\ast}$, the $k$-spectrum of $w$ is $\Sigma^{k}$ if and only 
if 
\[
\{\ta\tb,\tb\ta\}^k\cap\ScatFact_{2k}(w)\neq\emptyset.
\]
\end{theorem}
The previous theorem has an immediate consequence, which exactly characterises the weakly-$0$-balanced words of length $2k$ for which the maximal cardinality of $\ScatFact_k(w)$ is reached.
\refstepcounter{definition}

\begin{proof}
We will show this result by induction.
For $k=1$, the equivalence is:
\[
\ScatFact_1(w)=\Sigma\mbox{ iff }\{\ta\tb,\tb\ta\}\cap\ScatFact_2(w)\neq 
\emptyset.
\]
If both $\ta$ and $\tb$ are scattered factors of $w$, $\ta\tb$ or $\tb\ta$ has 
to be a factor and thus a scattered factor of $w$. On the other hand if $w$ has 
$\ta\tb$ or $\tb\ta$ as a scattered factor, it has $\ta$ and $\tb$ as scattered 
factors. 

Assume now that the equivalence holds for an arbitrary but fixed $k-1\in\N$. We will show it holds for $k$. 

For the $\Leftarrow$-direction consider 
$u\in\{\ta\tb,\tb\ta\}^k\cap\ScatFact_{2k}(u)$. Thus, 
$u\in\{\ta\tb,\tb\ta\}^{k-1}\{\ta\tb,\tb\ta\}$ and hence there exists 
$u'\in\{\ta\tb,\tb\ta\}^{k-1}$ with $u\in u'\{\ta\tb,\tb\ta\}$. By induction 
we have
$\ScatFact_{k-1}(u')=\Sigma^{k-1}$.  For any $x\in \Sigma^{k}$ exists 
$x'\in\Sigma^{k-1}$ 
with $x\in x'\{\ta,\tb\}$. This implies that there exist 
$a_0,\dots$, $a_{k-1}\in\Sigma^{\ast}$ with $u'=a_0x'[1]a_1\dots 
x'[k-1]a_{k-1}$ since $x'\in\ScatFact_{k-1}(u')$.
By 
\[
u\in a_0x'[1]a_1\dots x'[k-1]a_{k-1}\{\ta\tb,\tb\ta\} 
\]
it follows in both cases, namely $x=x'\ta$ or $x=x'\tb$, that 
$x\in\ScatFact_k(w)$.
This proves the inclusion $\Sigma^{k}\subseteq\ScatFact_k(w)$.
By $\ScatFact_k(w)\subseteq\Sigma^k$ the first direction is proven.

\medskip

For the $\Rightarrow$-direction assume $\ScatFact_{k}(w)=\Sigma^{k}$. Assume w.l.o.g. $w[|w|]=\ta$. Choose $x,y\in\Sigma^{\ast}$ 
with $w=xy$ and $x[|x|]=\tb$, and $y\in\ta^{\ast}$. As $\Sigma^{k-1}\tb\subset \Sigma^{k}$, it follows that  $\Sigma^{k-1}b\subseteq \ScatFact_{k}(x)$. Clearly, this means that 
 $\Sigma^{k-1}\subseteq \ScatFact_{k-1}(x[1..|x|-1])$. By the induction hypothesis, we get that $\{\ta\tb,\tb\ta\}^{k-1}\cap \ScatFact_{2(k-1)}(x[1..|x|-1]) \neq \emptyset$. Thus, 
 $\{\ta\tb,\tb\ta\}^{k-1}x[|x|]\ta \cap \ScatFact_{2k}(w[1..|x|+1]) \neq \emptyset$, because $w[1..|x|+1]=x[1..|x|]\tb$. Hence, $\{\ta\tb,\tb\ta\}^{k-1}\tb\ta \cap \ScatFact_{2k}(w) \neq \emptyset$. The conclusion follows.
\qed
\end{proof}

\setcounter{proposition}{\value{definition}}
\begin{proposition}\label{lem:full}\label{lemfull}
For $k\in\mathbb{N}_{\geq 3}$ and $w\in\Sigma_{\ed}^{2k}$ we have 
$w\in\{\ta\tb,\tb\ta\}^k$ if and only if 
$\ScatFact_k(w)=\Sigma^k$.
\end{proposition}
\refstepcounter{definition}

\begin{proof}
If $w\in\{\ta\tb,\tb\ta\}^k$, then 
$\{\ta\tb,\tb\ta\}^k\cap\ScatFact_{2k}(w)\neq\emptyset$ and the claim follows 
by Theorem~\ref{arbitraryw}. On the other hand if 
$\ScatFact_k(w)=\Sigma^k$ then 
$\{\ta\tb,\tb\ta\}^k\cap\ScatFact_k(w)\neq\emptyset$ and since $|w|=2k$ we get 
$w\in\{\ta\tb,\tb\ta\}^k$. \qed
\end{proof}

To see why from $w\in\{\ta\tb,\tb\ta\}^k$ it follows that $\ScatFact_k(w)=\Sigma^k$, note that, by definition, a word  $w\in\{\ta\tb,\tb\ta\}^k$ is just a concatenation of $k$ blocks from 
$\{\ta\tb,\tb\ta\}$. To construct the scattered factors of $w$, we can simply select from each block either the $\ta$ or the $\tb$. The resulting output is a word of length $k$, where in each position we could choose freely the letter. Consequently, we can produce all words in $\Sigma^k$ in this way. The other implication follows by induction.

Generalising Proposition~\ref{lemfull} for weakly-$c$-balanced words requires a more sophisticated approach. A 
generalisation would be to consider $w\in \{\ta\tb,\tb\ta\}^{k-c}\ta^c$. By Theorem~\ref{arbitraryw} we have $\ScatFact_{k-c}(w)=\Sigma^{k-c}$. But the 
size of $\ScatFact_{k-i}(w)$ for $i\in[c]_0$ depends on the specific choice 
of $w$. To see why, consider the  words $w_1= \tb\ta\ta\tb\tb\ta$ and 
$w_2=(\tb\ta)^3$. Then by Proposition~\ref{lemfull},
$|\ScatFact_3(w_1)|=8=|\ScatFact_3(w_2)|$. However, when we append an $\ta$ to the end of both $w_1$ and $w_2	$, we see that in fact
$|\ScatFact_4( w_1 \ta)|=11\neq12=|\ScatFact_4( w_2\ta )|$. 
The main difference 
between weakly-$0$-balanced and weakly-$c$-balanced words for $c>0$, regarding the maximum 
cardinality of the scattered factors-sets, comes from the role played by the factors $\ta^2$ and $\tb^2$ occurring in~$w$.  

In the remaining 
part of this section 
we present a series of results for weakly-$c$-balanced words. Intuitively, the words with many alternations between $\ta$ and $\tb$ have more distinct scattered factors. So, we will focus on such words mainly. 
Our first result is a direct consequence from Theorem~\ref{arbitraryw}. The second result concerns words avoiding $\ta^2$ and $\tb^2$ gives a method to identify efficiently the $\ell$-spectra of words which are prefixes of 
$(\ta\tb)^\omega$, for all $\ell$. Finally, we are able to derive a way to efficiently enumerate (and count) the scattered factors of length $k$ of $(\ta\tb)^{k-c}\ta^{c}$. 

\setcounter{corollary}{\value{definition}}
\begin{corollary}\label{corCmax1}
For $k\in\N_{\geq 3}$, $c\in[k]_0$, and 
$w\in\Sigma^{2k-c}$ weakly-$c$-balanced, the 
cardinality of $\ScatFact_{k-c}(w)$ is exactly $2^{k-c}$ if and only if  
$\ScatFact_{2(k-c)}(w)\cap \{\ta\tb,\tb\ta\}^{k-c}\neq\emptyset$.
\end{corollary}
\refstepcounter{definition}

\begin{proof}
The claim follows directly by Theorem~\ref{arbitraryw}.\qed
\end{proof}

As announced, we further focus our investigation on the words $w=(\ta\tb)^{k-c}\ta^c$. 
By Theorem~\ref{arbitraryw} we have $|\ScatFact_i(w)|=\Sigma^i$ for all $i\in[k-c]_0$. For all $i$ with $k-c<i\leq 
k$, a more sophisticated counting argument is needed. Intuitively, a scattered factor of length $i$ of $(\ta\tb)^{k-c}\ta^c$ consists of a part that is a scattered factor (of arbitrary length) of $(\ta\tb)^{k-c}$ followed by a (possibly empty) suffix of $\ta$s.
Thus, a full description of the $\ell$-spectra of words that occur as prefixes of $(\ta\tb)^\omega$, for all appropriate $\ell$, is useful. To this end, we introduce the 
notion of a deleting sequence: for a word $w$ and a scattered factor $u$ of $w$ the deleting sequence contains (in a strictly increasing order) $w$'s positions 
that have to be deleted to obtain $u$.

\begin{definition}\label{del-seq}
For $w\in\Sigma^{\ast}$, $\sigma=(s_1,\dots,s_\ell)\in[|w|]^\ell$, with $\ell\leq |w|$ and 
$s_i<s_{i+1}$ for all $i\in[\ell-1]$, is a {\em deleting sequence}.
The scattered factor~$u_\sigma$ associated to a deleting sequence $\sigma$ is $u_\sigma=u_1\dots u_{\ell+1}$, where $u_1= w[1..s_1-1]$, $u_{\ell+1}=w[s_\ell+1..|w|]$, and $u_i=w[s_{i-1}+1..{s_{i}}-1]$ for $2\leq i\leq \ell$.
Two sequences $\sigma, \sigma'$ with $u_\sigma=u_{\sigma'}$ are called~{\em equivalent}.
\end{definition}

For the word $w=\ta\tb\tb\ta\ta$ and $\sigma=(1,3,4)$ the associated scattered 
factor is $u_\sigma=\tb\ta$. Since $\tb\ta$ can also be generated by $(1,3,5)$, 
$(1,2,4)$ and $(1,2,5)$, these sequences are equivalent.

In order to determine the $\ell$-spectrum of a word $w\in\Sigma^{n}$ for 
$\ell,n\in\N$, we can determine how many equivalence classes does the equivalence defined above have, for 
sequences of length $k=n-\ell$. The following three lemmas characterise 
the equivalence of deleting sequences.

\setcounter{lemma}{\value{definition}}
\begin{lemma}\label{reduce}
Let $w\in\Sigma^n$ be a prefix of 
$(\ta\tb)^\omega$. Let $\sigma=(s_1,\ldots,s_k)$ be a deleting sequence for $w$ 
such that there exists $j\geq 2$ with $s_{j-1}<s_j-1$ and $s_j+1=s_{j+1}$. Then 
 $\sigma$ is equivalent 
$\sigma'=(s_1,\ldots,s_{j-1},s_j-1,s_{j+1}-1,s_{j+2},\ldots s_k)$, i.e., $\sigma'$ is the 
sequence $\sigma$ where both $s_j$ and $s_{j+1}$ were decreased~by~$1$.
\end{lemma}
\refstepcounter{definition}

\begin{proof}
Since $s_{j-1}<s_j-1$, the factor $u_\sigma$ contains the letter $w[s_j-1]$. If 
$w[s_j]=\ta$ then $w[s_{j+1}]=w[s_j+1]=\tb$ and $w[s_j-1]=\tb$. Clearly, when 
deleting $w[s_j-1]$ and 
$w[s_{j}]$ according to the sequence $\sigma'$, the $\tb$ that was corresponding 
to $w[s_j-1]$ will be replaced by a letter $\tb$ corresponding to $w[s_{j+1}]$, 
which is not deleted. So, in the end, $u_{\sigma'}=u_\sigma$. The case 
$w[s_j]=\tb$ is analogous. \qed
\end{proof}

\setcounter{lemma}{\value{definition}}
\begin{lemma}\label{normal-form}
Let $w\in\Sigma^n$ be a prefix of 
$(\ta\tb)^\omega$. Let $\sigma=(s_1,\ldots,s_k)$ be a deleting sequence for $w$. 
Then there exists an integer $j\geq 0$ such that $\sigma$ is equivalent to the 
deleting sequence $(1,2,\ldots,j,s'_{j+1},\ldots,s'_k)$, where $s'_{j+1}>j+1$ 
and $s'_i>s'_{i-1}+1$, for all $j<i\leq k$. Moreover, $j\geq 1$ if and only if $\sigma$ contained two consecutive positions or $\sigma$ started with $1$.
\end{lemma}
\refstepcounter{definition}

\begin{proof}
Let $\sigma_0=\sigma$. For $i\geq 0$, we iteratively transform $\sigma_i$ into 
$\sigma_{i+1}$ as follows: if $\sigma_i$ contains on consecutive positions the 
numbers $g,t,t+1,h$, such that $g<t-1$ and $h>t+2$, we replace them by 
$g,t-1,t,h$ and obtain the sequnce $\sigma_{i+1}$. By Lemma \ref{reduce}, 
$\sigma_i$ is equivalent to $\sigma_{i+1}$. It is clear that in $O(n^2)$ steps 
we will reach a sequence $\sigma_\ell$ which cannot be transformed anymore. We 
take $\sigma'=\sigma_\ell$ and it is immediate that it will have the required 
form.\qed
\end{proof}

\setcounter{lemma}{\value{definition}}
\begin{lemma}\label{unique-normal-form}
Let $w\in\Sigma^n$ be a prefix of 
$(\ta\tb)^\omega$. Let $\sigma_1=(1,2,\ldots,j_1,s'_{j_1+1}$, 
$\ldots,s'_k)$, 
where $s'_{j_1+1}>j_1+1$ and $s'_i>s'_{i-1}+1$, for all $j_1<i\leq k$, and 
$\sigma_2=(1,2,\ldots,j_2,s''_{j_2+1},\ldots,s''_k)$, where $s''_{j_2+1}>j_2+1$ 
and $s''_i>s''_{i-1}+1$, for all $j_2<i\leq k$. If $\sigma_1\neq \sigma_2$ then 
$\sigma_1$ and $\sigma_2$ are not equivalent (i.e., $u_{\sigma_1}\neq 
u_{\sigma_2}$).
\end{lemma}
\refstepcounter{definition}

\begin{proof}
We first consider the case $j_1=j_2$. Let $\ell$ to be minimum such that 
$s'_\ell\neq s''_\ell$. We can assume without losing generality that  $s'_\ell< 
s''_\ell$. Then $u_{\sigma_1}$ and $u_{\sigma_2}$ share the same prefix of 
length $t=(s'_\ell-1)-(\ell-1)$. This prefix ends with $w[s'_\ell-1]$ and is 
followed by $w[s'_\ell+1]$ in $u_{\sigma_1}$ and, respectively, by $w[s'_\ell]$ 
in $u_{\sigma_2}$. But $w[s'_\ell+1]\neq w[s'_\ell]$, so $u_{\sigma_1}\neq 
u_{\sigma_2}$.

Further, we consider the case when $j_1<j_2$ (the case $j_2<j_1$ is symmetric); 
assume, as a convention, that $s''_{k+1}=0$ and let $d=j_2-j_1$. Clearly, $j_1$ 
and $j_2$ must have the same parity, or $u_{\sigma_1}$ and $u_{\sigma_2}$ would 
start with different letters, so they would not be equal. Let $\ell$ to be 
minimum integer such that $s'_\ell-j_1\neq s''_{\ell+d}-j_2$; because 
$s''_{k+1}=0$ by convention, we have $\ell\leq k$. If both $\ell$ and $\ell+d$ 
are at most $k$, then we get similarly to the case $j_1=j_2$ that 
$u_{\sigma_1}\neq u_{\sigma_2}$. In the case when $\ell\leq k<\ell+d$, then, by 
length reasons, all positions $j>s_\ell$ (so, including $s_\ell+1$) in $w$ 
should belong to  $\sigma_1$, a contradiction. This concludes our proof. \qed
\end{proof}

Lemmas \ref{reduce}, \ref{normal-form}, and \ref{unique-normal-form} show that the representatives of the equivalence classes w.r.t. the equivalence relation between deleting sequences, introduced in Definition \ref{del-seq}, are the sequences $(1,2,\ldots,j,s'_{j+1},\ldots,s'_k)$, where $s'_{j+1}>j+1$ and $s'_i>s'_{i-1}+1$, for all $j<i\leq k$. For a fixed $j\geq 1$, the number of such sequences is ${{(n-j-1) -(k-j)+1}\choose{k-j}} = {{n-k}\choose{k-j}}$. For $j=0$, we have ${{(n-1)-k+1}\choose {k}}={{n-k}\choose {k}}$ nonequivalent sequences (note that none starts with $1$, as those were counted for $j=1$ already). In total, we have, for a word $w$ of length $n$, which is a prefix of $(\ta\tb)^\omega$, exactly $ \sum_{j\in [k]_0}{{n-k}\choose{k-j}}$ nonequivalent deleting sequences of length $k$, so $ \sum_{j\in [k]_0}{{n-k}\choose{k-j}}$ different scattered factors of length $n-k$. In the above formula, we assume that ${a\choose b}=0$ when $a<b$.

Moreover, the distinct scattered factors of length $\ell=n-k$ of $w$ can be obtained efficiently as follows. For $j$ from $0$ to $\ell$, delete the first $j$ letters of $w$. For all choices of $\ell-j$ positions in $w[j+1..n]$, such that each two of these positions are not consecutive, delete the letters on the respective positions. The resulted word is a member of $\ScatFact_\ell(w)$, and we never obtain the same word twice by this procedure. 
The next theorem follows from the above.

\setcounter{theorem}{\value{definition}}
\begin{theorem}\label{(ab)^k}
Let $w$ be a word of length $n$ which is a prefix of $(\ta\tb)^\omega$. Then $|\ScatFact_{\ell}(w)|=\sum_{j\in [n-\ell]_0}{{\ell}\choose{n-\ell-j}}$.
\end{theorem}
\refstepcounter{definition}

A straightforward consequence of the above theorem is that, if $\ell \leq 
n-\ell$ then $|\ScatFact_{\ell}(w)|=2^\ell$. 
With Theorem~\ref{(ab)^k}, we can now completely characterise the cardinality of the 
$\ell$-spectra of the weakly-$c$-balanced word $(\ta\tb)^{k-c}\ta^c$ for $\ell\leq k$.

\setcounter{theorem}{\value{definition}}
\begin{theorem}\label{maxcbalanced}
Let $w=(\ta\tb)^{k-c}\ta^c$ for $k\in\N$, $c\in[k]_0$. Then, for $i\leq k-c$ we have $|\ScatFact_{i}(w)|=2^i$. For $k\geq i>k-c$ we have $|\ScatFact_{i}(w)| = 1+2^{k-c} + \sum_{j\in [(i+c)-k-1]_0}|\ScatFact_{i-j-1}((\ta\tb)^{k-c-1}\ta)|$. 
\end{theorem}\refstepcounter{definition}

\begin{proof}
We will need to show the proof for $k\geq i>k-c$, as the other part follows immediately from Theorem \ref{arbitraryw}. 

We give a method to count the scattered factors of $w=(\ta\tb)^{k-c}\ta^c$. To 
begin with, we have the scattered factor $\ta^i$. All the other scattered 
factors must contain a letter $\tb$. Thus, we count separately the scattered 
factors of the form $u\tb\ta^j$, for each $j\in [i-1]_0$. This is equivalent to 
counting in how many ways we can choose $u$. For each such $u$ we will just have 
to append $\tb \ta^j$ at the end to get the desired scattered factors of length. 
Thus, $|u|=i-j-1$. If $j\geq c$ then $u$ should occur as a scattered factor of 
$(\ta\tb)^{k-j-1}\ta$ (in order to be able to append $\tb\ta^j$ at its end and 
still stay as a scattered factor of $w$), while if $j<c$ then $u$ should occur 
as a scattered factor of $(\ta\tb)^{k-c-1}\ta$. In the first case, the length of 
the scattered factor $u$ we want to generate is less than half of the length of 
the word $(ab)^ta$ from which we generate it. So, there are $2^{i-j-1}$ choices 
for $u$. In the second case, if $j\geq (i+c)-k$, again, the length of the 
scattered factor $u$ we want to generate is less than half of the length of the 
word $(ab)^{k-c-1}a$ from which we generate it. So, there are $2^{i-j-1}$ 
choices for $u$ again. Finally, if $j< (i+c)-k$, then there $i-j-1>k-c-1$, and 
we need Theorem \ref{arbitraryw} to generate $u$. There are 
$|\ScatFact_{i-j-1}((\ta\tb)^{k-c-1}\ta)|$ ways to choose $u$ in this case. 
Summing all these up, we get the result from the statement:
\[1 + \sum^{i-1}_{j=i+c-k}2^{i-j-1} + \sum_{j\in[i+c-k-1]_0}\ScatFact_{i-j-1}((\ta\tb)^{k-c-1}\ta) =\]
\[1 + 2^{k-c} +   \sum_{j\in[i+c-k-1]_0}\ScatFact_{i-j-1}((\ta\tb)^{k-c-1}\ta). \]
This concludes our proof.\qed
\end{proof}

As in the case of the scattered factors of prefixes of $(\ta\tb)^\omega$, we have a precise and efficient way to generate the scattered factors of $w=(\ta\tb)^{k-c}\ta^c$. For scattered factors of length $i\leq k-c$ of $w$, we just generate all possible words of length $i$. For greater $i$, on top of $\ta^i$, we generate separately the scattered factors of the form $u\tb\ta^j$, for each $j\in [i-1]_0$. It is clear that, in such a word, $|u|=i-j-1$, and if $j\geq c$ then $u$ must be a scattered factor of $(\ta\tb)^{k-j-1}\ta$, while if $j<c$ then $u$ must be a scattered factor of $(\ta\tb)^{k-c-1}\ta$. If $j\geq (i+c)-k$ then, by Theorem \ref{(ab)^k}, $u$ can take all $2^{i-j-1}$ possible values. For smaller values of $j$, we need to generate $u$ of length $i-j-1$ as a scattered factor of $(\ta\tb)^{k-c-1}\ta$, by the method described after Proposition~\ref{lemfull}. 

Nevertheless, Theorems \ref{(ab)^k} and \ref{maxcbalanced} are useful to see that in order to determine the cardinality of the sets of scattered factors of words consisting of alternating $\ta$s and $\tb$s or, respectively, of $(ab)^{k-c}a^c$, it is not needed to generate these sets effectively.

\section{Cardinalities of $k$-Spectra of Weakly-$0$-Balanced 
Words}\label{sec:card}
In the last section a characterisation for 
the smallest and the largest $k$-spectra of words of a given length are presented (Theorem~\ref{lemsmallest} and Proposition~\ref{lemfull}). In this 
section the part in between will be investigated for weakly-$0$-balanced words (i.e.\,words 
of length $2k$ with $k$ occurrences of each letter). 
As before, we shall assume that $k\in \N_{\geq 3}$. 
In the particular case that $k=3$, we have already proven that the $k$-spectrum with minimal cardinality 
has 
$4$ elements and 
that the maximal cardinality is $8$. Moreover as mentioned in 
Remark~\ref{smallestgap} a 
$k$-spectrum of cardinality $5$ does not exist for weakly-$0$-balanced words of length $2k$. The question remains
if $k$-spectra of cardinalities $6$ and $7$ exist, and if so, for which words.

Before showing that a $k$-spectrum of cardinality $2^k-1$ for 
weakly-$0$-balanced 
words of length $2k$ also exists for all 
$k\in\N_{\geq 3}$, we prove that only scattered factors of the form 
$\tb^{i+1}\ta^{k-i-1}$ for $i\in[k-2]_0$ (up to renaming, reversal) can be 
``taken out'' from the full set of possible scattered factors independently, without additionally requiring the removal of additional scattered factors as well.
In particular, if a word of length $k$ of another form is absent from the set 
of 
scattered factors 
of $w$, then  $|\ScatFact_k(w)|<2^{k}-1$ follows.

\setcounter{lemma}{\value{definition}}
\begin{lemma}\label{singleelement}
If for $w\in\Sigma_{\ed}^{2k}$ there exists $u\notin\ScatFact_k(w)$ with 
$u\notin\{\tb^{i}\ta^{k-i}\mid i\in[k-1]\} \cup \{ \ta^{i}\tb^{k-i}\mid i\in[k-1]\}$, then 
$|\ScatFact_k(w)|<2^k-1$.
\end{lemma}\refstepcounter{definition}

\begin{proof}
Let be $i\in[k-2]_0$. Consider firstly $u=\tb^r\ta^s$ for $r+s=k$ and 
$r\not\in[i]\cup\{k-i,\dots, 
k\}$ and $\Sigma^k\backslash\{u\}\supset\ScatFact_k(w)$ for a word 
$w\in\Sigma_{\ed}^{2k}$.
If $\tb^{r+1}\ta^{s-1}$ is also not a scattered factor of $w$, the claim is 
proven (in this case 
two elements of $\Sigma^k$ are missing in $\ScatFact_k(w)$). Assume 
$\tb^{r+1}\ta^{s-1}\in\ScatFact(w)$. This implies that (possibly intertwined) 
$(s-1)$ occurrences 
of $\ta$ follow $(r+1)$ occurrences of $\tb$. Since $u$ is not a scattered 
factor of $w$, after 
these $(s-1)$ $\ta$s only $\tb$s may occur. If $\tb^{r-1}\ta^{s}\tb$ is not a 
scattered factor, the
claim is again proven and so suppose that it is one. This implies that the 
$(r-1)$ $\tb$s are 
preceded by $\ta$s and not by $\tb$s. This implies that $\tb^{r+1}\ta^{s-1}$ is 
not a scattered 
factor and that contradicts the assumption. Consider now $u=u_1\tb^r\ta^s\tb^t 
u_2$ with $|u|=k$ 
not to be a scattered factor of $w$ for $r,s,t\in\N$. Following the same 
arguments as before, the 
claim is proven if $u_1\tb^{r-1}\ta^s\tb^{t+1}u_2$ is not a scattered factor 
and 
hence it is 
assumed to be one. This implies that exactly $|u_1|_\tb$ $\tb$s occur before 
$\tb^{r-1}$. This 
implies that $u_1\tb^{r+1}\ta^s\tb^{t-1}u_2$ is not a scattered factor of $w$ 
of 
length $k$. 
Analogously it can be proven that scattered factors containing the switch from 
$\ta$ to $\tb$ and 
back to $\ta$ cannot lead to the cardinality $2^k-1$.\qed
\end{proof}

\setcounter{proposition}{\value{definition}}
\begin{proposition}\label{lem2^k-1}
For $k\in\N_{\geq 3}$ and $w \in \Sigma^{2k}_{\ed}$, the set $\ScatFact_k(w)$ has
$2^k-1$ elements if and only if 
$w\in\{(\ta\tb)^{i}\ta^2\tb^2(\ta\tb)^{k-i-2}\mid i\in[k-2]_0\}$ (up 
to renaming and reversal). In particular
$\ScatFact_k(w)=\Sigma^k\backslash\{\tb^{i+1}\ta^{k-i-1}\}$ holds for 
$w=(\ta\tb)^{i}\ta^2\tb^2(\ta\tb)^{k-i-2}$ with $i\in[k-2]_0$.
\end{proposition}\refstepcounter{definition}

\begin{proof}
Let be $i\in[k-2]_0$.  First "$\Leftarrow$" will be proven 
and for that consider $w=(\ta\tb)^i\ta^2\tb^2(\ta\tb)^{k-i-2}$. By 
Lemma~\ref{lemfull} follows
\[
\ScatFact_i((\ta\tb)^i)=\Sigma^i\mbox{ and }
\ScatFact_{k-i-2}((\ta\tb)^{k-i-2})=\Sigma^{k-i-2}. 
\]
With 
$\ScatFact_2(\ta^2\tb^2)=\{\ta\ta,\ta\tb,\tb\tb\}$ the $k$-spec\-trum of $w$ 
has 
at least $3\cdot 
2^{i}\cdot 2^{k-i-2}=3\cdot 2^{k-2}=2^k-2^{k-2}$ elements. Notice that by this 
construction, 
scattered factors with  a $\tb\ta$ at the middle position cannot be 
reached. For this reason we have to have a look at $w$'s remaining scattered 
factors not 
being gained by the above construction. This means that not only $i$ letters 
are 
allowed to be 
taken of the first part and not only $k-i-2$ letters from the last part. 

Having a deeper look into 
$(\ta\tb)^i$ one can notice that all {\em binary numbers} 
(encoded by $\ta,\tb$) of length $i$ are scattered factors of  
$(\ta\tb)^{i-1}\ta$. Appending to these scattered factors a $\tb$ implies
that nearly all {\em binary numbers}  are in the 
$i+1$-spectrum of 
$\ta\tb^i$. Appending now an $\ta$ from the middle part and then each of the 
words from the last 
part leads to nearly all remaining scattered factors of the $k$-spectrum of 
$w$. 
The only missing 
word is $\tb^{i+i}$, since the last $\tb$ cannot be reached within the first 
part. This implies 
that the word $\tb^{i+1}\ta^{k-i-1}$ is not in the $k$-spectrum of $w$ since 
with the \nth{(i+1)}
$\tb$ the middle part is reached and the last part contains only $k-i-2$ 
$\ta$s. 
This concludes
$|\ScatFact_k(w)|=2^k-1$.

On the other hand if $|\ScatFact_k(w)|=2^k-1$ an element of the form 
$\tb^{i+1}\ta^{k-i-1}$ for an 
$i\in[k-2]_0$ is missing in the $k$-spectrum of $w$. Moreover this is exactly 
the only element 
missing. Fix an $i\in[k-2]_0$ and set $u=\tb^{i+1}\ta^{k-i-1}$. The proof will 
be very technically 
and exclude step by step all other possibilities than $w$ being 
$(\ta\tb)^i\ta^2\tb^2(\ta\tb)^{k-i-2}$. Firstly consider $i=k-2$. This implies 
$u=\tb^{k-1}\ta$.
In this case $w$ has to end in $\tb^2$ but not in $\tb^3$ since otherwise 
$\tb^{k-2}\ta^2$ would 
not be a scattered factor. If $w$ were of the form $w_1\tb\ta\tb^2$, 
$|w_1|_a=k-1$ and 
$|w_1|_b=k-3$ would hold which would imply that $\tb^{k-2}\ta^2$ is not a 
scattered factor. If $w$ 
ended in $\ta^3\tb^2$, $\ta^{k-2}\tb\ta$ would be excluded. Hence, $w$ ends in 
$\ta^2\tb^2$. 
Suppose at last that $w=(\ta\tb)^{\ell}\ta^2\tb^2w_2$ holds for $\ell<k-2$ and 
$w_2\in\Sigma^{\ast}$. Then $w_2$ has each $(k-\ell-2)$ $\ta$ and $\tb$. Thus 
$\tb^{\ell+1}\ta^{k-\ell-1}$ is not a scattered factor of length $k$. This 
proofs that for $i=k-2$
$w=(\ta\tb)^{k-2}\ta^2\tb^2$ is implied by $\tb^{k-2}\ta^{1}$ being the only 
excluded scattered 
factor from $\Sigma^k$. Hence assume $i\in[k-3]_0$.\\
{\tt Supposition:} $w$ ends in $\tb^{\ell}$ for $\ell\geq 2$\\
If $i<k-2$ holds, then $\tb^{k-1}\ta\not\in\ScatFact_k(w)$ follows and since 
$i+1<k-1$ holds, this
element is different from $u$.  \\
In the next step it will be shown that exactly $k-i-2$ repetitions of $\ta\tb$ 
are a suffix of 
$w$.\\
{\tt Supposition:} $w=w_1\tb^2(\ta\tb)^{\ell}$\\
If $\ell>k-i-2$ held, $\tb^{i+1}\ta^{k-i-1}$ would not be a scattered factor of 
$w$. If 
$\ell<k-i-2$ held, $\tb^{k-\ell-1}\ta^{\ell+1}$ would not be a scattered factor 
since $w_1$ has 
$(k-1)$ $\ta$ and $(k-\ell-2)$ $\tb$.\\
{\tt Supposition:} $w=w_1\ta^2(\tb\ta)^{\ell}b$\\
In this case $|w_1|_\ta=k-2-\ell$ and $|w_1|_\tb=k-\ell-1$ holds. This implies 
that 
$\ta^{k-2-\ell}\tb^{\ell+1}\ta$ is not in the $k$-spectrum of $w$.\\
Consequently there exists a $w_1$ such that $w=w_1\tb^2(\ta\tb)^{k-i-2}$ holds. 
In the next it 
will 
be shown that $\tb^2$ has to be preceded by $\ta^2$.\\
{\tt Supposition:} $w=w_1\tb^3(\ta\tb)^{k-i-2}$\\
Here $w_1$ has $(i+2)$ $\ta$ and $(i-1)$ $\tb$ and hence 
$\tb^i\ta^{k-i-2}\tb^2$ 
is not a 
scattered 
factor of length $k$ of $w$.\\
{\tt Supposition:} $w=w_1\tb\ta\tb^2(\ta\tb)^{k-i-2}$\\
This implies $\ta^{i+2}\tb\ta\tb^{k-i}\not\in\ScatFact_k(w)$ since $w_1$ has 
$i+1$ occurrences of 
$\ta$ and $i-1$ occurrences of $\tb$.\\
This proofs that $\ta^2\tb^2(\ta\tb)^{k-i-2}$ is a suffix of $w$. The case that 
this is preceded 
by 
another $\ta$ is excluded since then $\ta^i\tb\ta^{k-i-1}$ would not be in the 
$k$-spectrum of 
$k$. 
In the last step it will be shown that the first occurrence of $\ta^2$ is at 
the 
point $2\ell$.\\
{\tt Supposition:} $w=(\ta\tb)^{\ell}\ta^2w_2$ for $\ell\neq i$\\
If $\ell$ is smaller than $i$, $|w_2|_\ta=k-\ell-2$ and $|w_2|_\tb=k-\ell$ hold 
and 
$\tb^{\ell+1}\ta^{k-\ell-1}\not\in\ScatFact_k(w)$ follows. If $\ell$ is greater 
than $i$, in 
contradiction to the main assumption $\tb^{i+1}\ta^{k-i-1}$ is a scattered 
factor, because 
$\tb^{i+1}$ is a scattered factor of $(\ta\tb)^{\ell}$ and 
$k-\ell+\ell-(i+1)=k-i-1$ $\ta$ are 
left 
in the rest of $w$. \\
Combining $w=(\ta\tb)^i\ta^2w_2$ and $w=w_1\ta^2\tb^2(\ta\tb)^{k-i-2}$ the 
claim 
that $w$ is of the 
form $(\ta\tb)^i\ta^2\tb^2(\ta\tb)^{k-i-2}$ is proven.\qed
\end{proof}

By Proposition~\ref{lem2^k-1} we get that $7$ is a possible cardinality of the set of scattered factors of length $3$ of  
weakly-$0$-balanced words of length $6$ and, moreover, that exactly the words 
$\ta^2\tb^2\ta\tb$ and $\ta\tb\ta^2\tb^2$ (and symmetric words obtained by 
reversal and renaming) have seven different scattered factors.
The following theorem demonstrates that there always exists a weakly-$0$-balanced 
word $w$ of length $2k$ such that $|\ScatFact_k(w)| = 2k$. Thus, for the case $k=3$ also the 
question if six is a possible cardinality of $\ScatFact_3(w)$ can be answered positively.

\setcounter{theorem}{\value{definition}}
\begin{theorem}\label{lem2k}
The $k$-spectrum of a word $w\in\Sigma_{\ed}^{2k}$ has exactly $2k$ elements if 
and only if
$w\in\{\ta^{k-1}\tb\ta\tb^{k-1},\ta^{k-1}\tb^k\ta\}$ holds (up to renaming and 
reversal). Moreover, 
there does not exist a 
weakly-$0$-balanced word $w\in\Sigma_{\ed}^{2k}$ with a $k$-spectrum of 
cardinality $2k-i$ for $i\in[k-2]$.
\end{theorem}\refstepcounter{definition}

\begin{proof}
Consider first $w=\ta^{k-1}\tb\ta\tb^{k-1}$. Since the $k$-spectrum of 
$\ta^k\tb^k$ is a subset of the $k$-spectrum of 
$w$, the $k$-spectrum of $w$ has at least 
$k+1$ elements.
Additionally $w$
has the scattered factors of the form $\ta^i\tb\ta\tb^{k-2-i}$, which sum up to 
$k-1$. Hence
$|\ScatFact_k(w)|=k+1+k-1=2k$ holds. Moreover
$\ta^{k-1}\tb^k\ta$ has
all elements of $\ta^k\tb^k$'s $k$-spectrum as scattered factors. Here the word 
has in 
addition 
all words of the form $\ta^i\tb^{k-1-i}\ta$ as scattered factors which sum up to 
$k-1$ as well. 
This proves that both words have a scattered factor set of cardinality $2k$.

The other direction will be proven by contraposition following the two main 
cases 
\[
\ta^{k-1}\tb\ta\tb^{k-1}\quad\mbox{and}\quad\ta^{k-1}\tb^k\ta.
\]

Assume first $w=\ta^{\ell}\tb x$ for $\ell\in[k-2]_{\geq 2}$. Notice that it 
does not have to be 
considered that the word starts with one $\ta$, since this is 
symmetric to the reversal of the case $\ta^{k-1}\tb^k\ta$. This implies 
$|x|_a=k-\ell$ and 
$|x|_b=k-1$. Notice here $k-\ell<k-1$. Thus, there exists  a scattered factor 
$x'$ of $x$ of 
length 
$2(k-\ell)$ with $|x'|_a=|x'|_b=k-\ell$. By Lemma~\ref{lemsmallest} follows
\[
|\ScatFact_{k-\ell}(y)|=k-\ell+1\Leftrightarrow 
y\in\{\ta^{k-\ell}\tb^{k-\ell},\tb^{k-\ell}\ta^{k-\ell}\}
\]
and $|\ScatFact_{k-\ell}(y)|>k-\ell+1$ otherwise.
This implies that the $(k-\ell)$-spectrum of $x'$ is minimal with respect 
to cardinality if $x'$ is either 
$\ta^{k-\ell}\tb^{k-\ell}$ or $\tb^{k-\ell}\ta^{k-\ell}$. For giving a lower 
bound of the 
cardinality of $w$'s scattered factor set of length $k$, it is sufficient to 
only take these both 
options into consideration. This implies that it is not necessary to examine the 
cases where
$x$ contains other scattered factors with both $k-\ell$ $\ta$ and $\tb$.

\noindent{\tt case 1:} $x'=\ta^{k-\ell}\tb^{k-\ell}$\\
Thus $x$ contains $\ell-1$ $\tb$ which are not in $x'$.\\
{\tt case a:} $x=\tb^{\ell-1}\ta^{k-\ell}\tb^{k-\ell}$\\
In this case $w=\ta^{\ell}\tb^{\ell}\ta^{k-\ell}\tb^{k-\ell}$ holds and 
that the $k$-spectrum of $\ta^k\tb^k$ is a subset of $\ScatFact_k(w)$ follows. 
\\
{\tt case i:} $\ell<k-\ell$\\
For all $s\in[\ell]$ the words $\ta^{\ell-s}\tb^s\ta^{k-\ell},\dots, 
\ta^{\ell}\tb^s\ta^{k-\ell-s}$
are well-defined and sum up to $s+1$. Moreover for every $s_2\in[k-\ell]$ exists 
$r_1\in\N_0$
and exist $r_2,s_2\in\N$ such that the words 
$\ta^{r_1}\tb^{s_1}\ta^{r_2}\tb^{s_2}$
with $s_1+r_1+s_2+r_2=k$ are all distinct and distinct to the aforementioned. 
Thus, in this case
\[
k+1+\sum_{s=1}^{\ell}(s+1)+k-\ell=2k+1-\ell+\frac{\ell(\ell+1)}{2}+\ell\geq 2k+4
\]
is a lower bound for $\ScatFact_k(w)$.\\
{\tt case ii:} $\ell>k-\ell$\\
Consider here for $r\in[k-\ell]$ the words 
$\tb^{\ell-r}\ta^r\tb^{k-\ell},\dots,\tb^{\ell}\ta^r\tb^{k-\ell-r}$. For fixed 
$r$ these are $r+1$.
Moreover in this case for all $r_1\in[\ell]$ exist $s_1,r_2\in\N$ and $s_2\in\N$ 
such that the 
words
$\ta^{r_1}\tb^{s_1}\ta^{r_2}\tb^{s_2}$ with $s_1+r_1+s_2+r_2=l$ are all distinct 
and distinct
to the aforementioned.  In total this sums up to
\[
k+1+\sum_{r=1}^{k-\ell}(r+1)+\ell=k+1+\frac{(k-\ell)(k-\ell+1)}{2}+(k-\ell)+\ell
\geq 2k+4
\]
different scattered factors.\\
{\tt case b:} $x=\ta^{k-\ell}\tb^{k-1}$\\
Thus, $w=\ta^{\ell}\tb\ta^{k-\ell}\tb^{k-1}$ holds. Here it holds as well that 
the $k$-spectrum of 
$\ta^k\tb^k$ is a subset of $\ScatFact_k(w)$. Moreover all words of the
form $\tb\ta^r\tb^s$ for $r+s=k-1$ and $r\in[k-\ell]$ are different scattered 
factors, i.e. 
$k-\ell$ many. Additionally the words $\ta^r\tb\ta\tb^s$ for $r+s=k-2$ and 
$r,s>0$ are different
scattered factors and distinct to the aforementioned. This sums up to 
$k+1+k-1+k-2=3k-2$ for the
cardinality of $\ScatFact_k(w)$. This proves the claim for $k\geq 3$.

\noindent{\tt case 2:} $x'=\tb^{k-\ell}\ta^{k-\ell}$\\
Consequently 
$x\in\{\tb^{k-1}\ta^{k-\ell},\tb^{k-\ell}\ta^{k-\ell}\tb^{\ell-1}\}$ holds.\\
{\tt case a:} $x=\tb^{k-1}\ta^{k-\ell}$\\
Hence $w=\ta^{\ell}\tb^{k}\ta^{k-\ell}$. Here only $\ell+1$ different scattered 
factors are of the 
form $\ta^r\tb^s$ exist and $k-\ell$ of the form $\tb^s\ta^r$ with $r+s=k$ 
(notice that the latter 
ones are only $k-\ell$ since among all of them one is in common with the first 
ones). Finally 
consider the words of the form $\ta^{r_1}\tb^s\ta^{r_2}$ with $r_1+r_2+s=k$ and 
$r_1,r_2,s>0$. This 
sums up to $\ell+1+k-\ell+k$. By $\ta^k\in\ScatFact_k(w)$, $|\ScatFact_k(w)|\geq 
2k+2$ follows.\\
{\tt case b:} $x=\tb^{k-\ell}\ta^{k-\ell}\tb^{ \ell-1}$\\
In this case $w=\ta^{\ell}\tb^{k-\ell+1}\ta^{k-\ell}\tb^{\ell-1}$ holds. Here 
the cardinality
of the $k$-spectrum of $w$ is determined analogously to case 1a.\qed
\end{proof}

By Proposition~\ref{lem2^k-1} and Theorem~\ref{lem2k} the possible cardinalities of $\ScatFact_3(w)$ for weakly-$0$-balanced words $w$ of length $6$ are completely 
characterized. Theorem~\ref{lem2k} determines the first gap in the set of cardinalities of $|\ScatFact_k(w)|$ for $w\in\Sigma_{\ed}^{2k}$: there does not exist 
a word $w\in\Sigma_{\ed}^{2k}$ with $|\ScatFact_k(w)|=k+i+1$ for $i\in[k-2]$ and 
$k\geq 3$, since all
words that are not of the form $\ta^k\tb^k$, $\tb^k\ta^k, 
\ta^{k-1}\tb\ta\tb^{k-1}$, or 
$\ta^{k-1}\tb^k\ta$ have a scattered factor set of cardinality at least 
$2k+1$. As the size of this first gap is linear in $k$, it is clear that the larger $k$ is, the 
more unlikely it is to find a $k$-spectrum of a small cardinality. 

In the following we will prove that the cardinalities $2k+1$ up to $3k-4$ are 
not reachable, i.e. $3k-3$ is the thirst smallest cardinality after $k+1$ and 
$2k$ (witnessed by, e.g.\,$\ta^{k-2}\tb^k\ta^2$).

\setcounter{lemma}{\value{definition}}
\begin{lemma}\label{gensquares}
For $i\in\big[\lfloor\frac{k}{2}\rfloor\big]$ and $j\in[k-1]$
\begin{itemize}
\item $|\ScatFact_k(\ta^{k-i}\tb^k\ta^i)|=k(i+1)-i^2+1$ for $k\geq 4$,
\item $|\ScatFact_k(\ta^{k-1}\tb^2\ta\tb^{k-2})|=3k-2$,
\item $|\ScatFact_k(\ta^{k-2}\tb^j\ta\tb^{k-j}\ta)|=
k(2j+2)-6j+2$ for $k\geq 5$, and 
\item $|\ScatFact_k(\ta^{k-2}\tb^j\ta^2\tb^{k-j})|=
k(2j+1)-4j+2$.
\end{itemize}
\end{lemma}\refstepcounter{definition}

\begin{proof}
For the first claim, let be $i\in \big[\lfloor\frac{k}{2}\rfloor\big]_{\geq 2}$.
The $k$-spectrum of $\ta^{k-i}\tb^k\ta^i$ contains exactly all words of the form 
$\ta^r\tb^s\ta^t$ 
with $r+s+t=k$, $t\in[i]_0$, $r\in[k-i]_0$, and $s\in[k]_0$. If $t$ and $r$ are 
fixed, $s$ is 
uniquely determined. Since all these scattered factors are different, the 
$k$-spectrum has
$(i+1)(k-i+1)=k(i+1)-i^2-1$ elements. Thus the first claim is proven.

\bigskip

For the second claim, notice that the scattered factors of $\ta^{k-1}\tb^2\ta\tb^{k-2}$ are of four different forms: 
$\tb^r\ta\tb^t$, 
$\ta^r\tb^s\ta$, $\ta^r\tb^s$, and $\ta^r\tb^{s_1}\ta\tb^{s_2}$. Notice that all 
these scattered
factors are different if in the second one $s$ is chosen greater than or equal 
to $1$ and in the 
last one $r,s_1,s_2\geq 1$ holds. The first and second one lead to two scattered 
factors, since 
for every $s\in[2]$ there are enough $\ta$ at the beginning for padding from the 
left. 
The third form leads to $k+1$ different
scattered as shown in Theorem~\ref{lemsmallest}. The last one is a little bit more 
complicated. 
Notice firstly that $r$ is at most $k-3$ since $s_1,s_2>0$ holds. In this case 
there exists only 
one possibility for chosing $s_1$ and $s_2$, namely as $1$. If $r$ is $k-4$ 
there exist two 
possibilities, namely $s_1=1$ and $s_2=2$ or vice versa. For $r\in[k-5]$ there 
exist always $2$ 
possibilities for the $\tb$s between the $\ta$s. This leads to $2(k-5)$ 
possibilities. Allover
it sums up to $2+2+k+1+1+2+2(k-5)=8+3k-10=3k-2$.

\bigskip

As in the proof of the second part, for the remaining parts  the scattered 
factors can be 
categorized in the form 
$\ta^r\tb^s$, $\tb^{r_1}\ta^s\tb^{r_2}$, $\ta^{r_1}\tb^s\ta^{r_2}$, and 
$\ta^{r_1}\tb^{s_1}\ta^{r_2}\tb^{s_2}$, where with appropriate chosen exponents 
no factors is 
counted twice. Also as before, $i$ can be chosen in 
$\left[\lfloor\frac{k}{2}\rfloor\right]$, 
since 
otherwise the proof is analogous for $k-i$. The first form contributes $k+1$ 
elements. The second 
and third form contribute $2i$ each, since $s$ resp. $r_2$ range in $[2]$. For 
the last form a 
distinction is necessary. If $r=k-3$ holds, $\ta^{k-3}\tb\ta\tb$ is the only 
scattered factor. If 
$r$ is smaller than $k-3$, $2i$ possibilities for each $r\in[k-3]$ lead to 
scattered factors. 
Allover this sums up to $k+1+2i+2i+1+2i(k-4)=k(2i+1)-4i+2$. By this the first 
claim is proven.

For the second claim again scattered factors of different forms will be 
distinguished. Since also 
here the minimal $k$-spectrum is a subset of the $k$-spectrum of $w$, these 
$k+1$ elements
counts for the cardinality. There exists $i$ many scattered factors of the form 
$\ta^r\tb^s\ta^2$
and $k-2$ of the form $\ta^r\tb^s\ta$, since with the last $\ta$ all occurrences 
of $\tb$ are 
before it. Assuming w.l.o.g. again that $i$ is at most $\frac{k}{2}$ only 
$\tb^{k-1}\ta$ is a 
scattered factor of the form $\tb^s\ta^r$. The scattered factors of the form 
$\tb^{r_1}\ta\tb^{r_2}\ta$ contribute $i$ many. The remaining two forms need 
again a case analysis.
There exists exactly one scattered factor of the form 
$\ta^{r}\tb^{s_1}\ta\tb^{s_2}$ for $r=k-3$
and exactly one scattered factor of the form $\ta^{r_1}\tb^{s_1}\ta\tb^{s_2}\ta$ 
for $r_1=k-4$. If
$r$ resp. $r_1$ are smaller there exists $i$ different scattered factors for 
each choice of 
$r\in[k-4]$ resp. $r_1\in[k-5]$. This sums up to 
$k+1+k-2+i+i+1+i+1+i(k-5)+1+k(i-4)= 2k+2+3i+ ik-5i+ ik-4i=k(2+2i)-6i+2$. \qed
\end{proof}

Notice that for $i\in\left[\lfloor\frac{k}{2}\rfloor\right]$ the sequence 
$(k(2i+1)-4i+2)_i$ is
increasing and its minimum is $3k-2$ while for 
$i\in\left[\lfloor\frac{k}{2}\rfloor\right]$ the 
sequence $(k(2i+2)-6i+2)_i$ is
increasing and its minimum is $4k-4$. The following lemma only gives lower 
bounds for specific forms of words, 
since, on the one hand, it proves to be sufficient for the Theorem~\ref{2gap} which describes the second gap, and, on 
the other hand, the proofs show that the formulas describing the exact number of scattered 
factors of a specific form are getting more and more complicated. It has to be 
shown that also words starting with $i$ letters $\ta$, for $i\in[k-3]$, 
have a $k$-spectrum of greater
(as lower is already excluded) cardinality.  By
Lemma~\ref{gensquares} only words with another transition from $\ta$'s 
to $\tb$'s need to be considered, ($w=\ta^{r_1}\tb^{s_1}w_1\ta^{r_1}\tb^{s_2}$).
W.l.o.g. we can assume $s_1$ to be maximal, such that $w_1$ starts with an 
$\ta$, 
and similarly, by 
maximality of $r_2$, ends with a $\tb$, thus only 
words of 
the form 
$\ta^{r_1}\tb^{s_1}\dots \ta^{r_n}\tb^{s_n}$ have to be considered, and by 
Proposition~\ref{lem:full}, it is sufficient to investigate $n<k$.

\setcounter{lemma}{\value{definition}}
\begin{lemma}\label{lemhelp3}
\begin{itemize}
\item $|\ScatFact_k(\ta^{k-2}\tb^i\ta\tb^j\ta\tb^{k-i-j})|\geq 
 3k-3$ for $i,j\in[k-2]$, $i+j\leq k-1$, 
 \item 
$|\ScatFact_k(\ta^{k-2}\tb^{s_1}\ta^{r_1}\tb^{s_2}\ta^{r_2}\tb^{s_3})|\geq 
3k-4$ for $s_1+s_2+s_3=k$, $r_1+r_2=2$,$s_1>0$, $r_1,r_2,s_2,s_3\geq 0$,
\item $|\ScatFact_k(  \ta^{r_1}\tb^{s_1}\dots\ta^{r_n}\tb^{s_n} )|\geq 3k-3$
for $r_1\leq k-3$, 
$\sum_{i\in[n]}r_i=\sum_{i\in[n]}s_i=k$, and $r_i,s_i\geq 1$.
\end{itemize}
\end{lemma}\refstepcounter{definition}

\begin{proof}
For the first claim, choose $i,j\in[k-2]$. Then all words of the form $\ta^r\tb^s$ for $r,s\in[k]_0$ 
are scattered 
factors of $w_{ij}$ and by Lemma~\ref{lemsmallest} follows that $w_{ij}$ has 
$k+1$ scattered 
factors of this form. Scattered factors of the form $\ta^{r_1}\tb^s\ta^{r_2}$ 
can occur in three 
variants. In the first variant only the second block of $\ta$ is involved after 
the first block of 
$\tb$, namely the second single $\ta$ is not involved. Since $i\in[k-2]$ holds, 
for each $s\in[i]$ 
exists $r_1,r_2$ ($r_2=1$) such that $\ta^{r_1}\tb^s\ta^{r_2}$ is a scattered 
factor of $w_{ij}$, 
i.e. $w_{ij}$ has additionally $i$ scattered factors. The second variant uses 
the $\ta$ of each 
the 
second and the third $\ta$-block. This only scattered factors of the form 
$\ta^{r_1}\tb^s\ta^{r_2}$
are of interest, the second $\tb$-block is not involved. If $i+j=k-1$ holds only 
$i-1$ scattered 
factors of this form occurs, otherwise again $i$ new elements are in the 
$k$-spectrum. If only the 
$\ta$ from the third block is involved then $j$ (resp. $j-1$) new elements are 
in the spectrum.
This sums up to at least $2i+j-2$ elements of the form 
$\ta^{r_1}\tb^s\ta^{r_2}$. A similar 
distinction leads to the number of scattered factors of the form 
$\ta^{r_1}\tb^{s_1}\ta^{r_2}\tb^{s_2}$. Assume first $r_2=1$ and for this only 
the $\ta$ from the 
second $\ta$-block. This implies that either only $\tb$ from the second block or 
from the second 
and third block can be taken for the last $\tb$-block in the scattered factor. 
Moreover 
$r_1,s_1,s_2$ are at most $k-3$. For each choice of $r_1$ in $[k-3]$ there are 
$\min\{j,k-2-i\}$
possibilities, which leads to 
\[
i\left((k-j-2)j+\sum_{\ell=1}^{k-j-2}k-2-\ell\right)
=6i+1\frac{1}{2}k^2i-kji-3\frac{1}{2}ki+1\frac{1}{2}j^2i+1\frac{1}{2}ji.
\]
If $\tb$ from the second and third block are allowed, all of the second block 
have to occur for 
obtaining different scattered factors to the previous ones. Thus, 
\begin{multline*}
i\left((k-j-i-2)j+\sum_{\ell=1}^{k-j-i-2}k-2-\ell\right)\\
=kij+\frac{1}{2}k^2-1\frac{1}{2}k-ik-jk-ij^2-i^2j-ij+1\frac{1}{2}i+1\frac{1}{2}
j+\frac{1}{2}
i^2+\frac { 1 }{2}j^2.
\end{multline*}
If both, the second and the third $\ta$-block, are involved 
$ik-1\frac{1}{2}i^2-ij-\frac{1}{2}i$
additional scattered factors are in the $k$-spectrum. This all sums up to
\begin{align*}
k+1+9i-2+1\frac{1}{2}k^2i-3\frac{1}{2}ki+\frac{1}{2}j^2i-\frac{1}{2}ji+\frac{1}{
2}j+\frac{1}{2}
i^2+\frac{1}{2}j^2.
\end{align*}
Since either $i^2\geq ij$ or $j^2\geq ij$ and $i,j\in[k-3]$ hold, this is 
greater than or equal to
\[
1\frac{1}{2}k^2-2\frac{1}{2}k+9\frac{1}{2}\geq 3k-3.
\]
Notice that additionally there exist scattered factors of other forms, which 
enlarge the concrete 
$k$-spectrum. 

\bigskip

For the second claim, consider first the case, when $s_2=0$, $r_1=0$, or $r_2=0$. This leads to words 
of the form
matching Lemma~\ref{gensquares} and consequently the $k$-spectrum has 
$k(2i+1)-4i+2\geq 
3k-2>3k-4$ 
elements. Consider now the case that $s_3=0$ holds and all other exponents are 
at least $1$. By 
Lemma~\ref{gensquares} follows again that each such word has at least 
$k(2i+2)-6i+2\geq 
4k-4>3k-4$ 
elements.
Finally by Lemma~\ref{lemhelp3} follows that the remaining words of the given 
form have at least 
$3k-3$ scattered factors.

\bigskip

Finally notice that $\ta^k$ is a scattered factor and $\ta^{k-i}\tb^i$ for $s_n$ also. 
Notice here, that
the proof leads to $s_{n-1}$ scattered factors, if in the claim $s_n=0$ would be 
allowed. Consider 
now the scattered factors of the form $\ta^i\tb^j$ for $i,j\in[k]$. Let $m$ be 
the number of the 
block in which the \nth{i} $\ta$ occurs. If $s_m+\dots+s_n\geq k-i$ holds, 
$\ta^i\tb^{k-i}$ is
a scattered factor of $w$. Consider the opposite. This implies that from the 
\nth{m} till the 
\nth{n} block less then $k-i$ $\tb$ occur. Thus in the blocks $1$ to $i$ there 
occur more than $i$ 
$\tb$. Since the \nth{i} $\ta$ is in the \nth{m} block, from this point till the 
end there are 
$k-i$ $\ta$. Hence $\tb^{i}\ta^{k-i}$ is a scattered factor of $w$. So in each 
case at least one 
scattered factor occurs, i.e. at least $k+1$ scattered factors of this form are 
in the 
$k$-spectrum.
Notice here, that the argument holds still if $s_m=0$ is allowed. With a similar 
argumentation the 
number of occurrences of the form $\ta^{i}\tb^{j}\ta^{k-i-j}$ will be shown. If 
for a specific 
$i,j$-combination $\ta^i\tb^j\ta^{k-i-j}$ is not a scattered factor, then choose 
$m_1, m_2$ such 
that the \nth{i} $\ta$ is in block $m_1$ and the \nth{j} $\tb$ after that is in 
block $m_2$. Thus
in the blocks $m_2+1$ to $n$ are less than $k-i-j$ $\ta$. Let $r_{m_1}'$ be the 
$\ta$ in the 
\nth{$m_1$} block which don't belong to $\ta^i$. Then $r_{m_1}'+\dots+r_{m_2}$ 
contains more than 
$k-j$ $\ta$ since $k-j-i$ $\ta$ occur in the \nth{$m_1$'} to the \nth{n} block. 
Thus 
$\ta^{r_{m_1}'}\tb^{s_{m_1}}\dots \ta^{r_{m_2}}\tb^{s_{m_2}'}$ is a scattered 
factor of length at 
least $k+1$ where $s_{m_2}'$ describes the part of the \nth{$m_2$} block until 
the \nth{j} $\tb$. 
If $1<m_1,m_2<n$ holds, $\tb\ta^{k-j-1}\tb^{j-2}$ is a scattered factor of $w$. 
If $m_1=m_2=1$ 
holds, $\ta^{k-j-3}\tb\ta\tb$ is a scattered factor. If both are equal to $n$, 
$\tb\ta^{k-j-1}\tb^{j-2}$ is a scattered factor. In both cases the last $\tb$ 
exist even if 
$s_m=0$ 
holds, since the scattered factor ends in the examined block $m_2$. If $m_1<m_2$ 
holds, there 
exists a factor of length $>k$ which can be narrowed to a factor starting in 
$\ta$, ending in 
$\tb$, and having at least one {\em switch} from $\tb$ back to $\ta$ and back to 
$\tb$. This 
concludes to at least $(k-2)^2$ scattered factors of the form 
$\ta^i\tb^j\ta^{k-i-j}$ (or a 
different one in exchange). By $k^2-k+3\geq 3k-3$ for $k\geq 5$ follows the 
claim.\qed
\end{proof}

By Lemmas~\ref{gensquares} and \ref{lemhelp3} we are able to prove the 
following theorem, which 
shows the second gap in the set of cardinalities of $\ScatFact_k$ for words in $\Sigma_{\ed}^{2k}$.

\setcounter{theorem}{\value{definition}}
\begin{theorem}\label{2gap}
For $k\geq 5$ there does not exist a word $w\in\Sigma_{\ed}^{2k}$ with 
$k$-spectrum of cardinality 
$2k+i$ for $i\in[k-4]$. In other words, i.e. between $2k+1$ and $3k-4$ is a 
cardinality-gap.
\end{theorem}\refstepcounter{definition}

\begin{proof}
Theorems~\ref{lemsmallest} and \ref{lem2k} show that exactly the words 
$\ta^k\tb^k$, 
$\ta^{k-1},\tb\ta\tb^{k-1}$, and $\ta^{k-1}\tb^k$ $\ta$ have $k$-spectra of 
cardinality less than 
or 
equal to $2k$. By Lemma~\ref{gensquares} and \ref{lemhelp3} follows that 
$\ta^{k-2}\tb^k\ta^2$ has a 
$k$-spectrum of 
cardinality $3k-3$. Assume a 
$w\in\Sigma_{\ed}^{2k}\backslash\{\ta^k\tb^k,\ta^{k-1}\tb\ta$ 
$\tb^{k-1},\ta^{k-1}\tb^k\ta,\ta^{k-2}
\tb^k\ta^2\}$. Since renaming and reversal do not influence the cardinality, it 
can be assumed 
w.l.o.g. that $w$ starts with $\ta$. By assumption $w$ does not start with 
$\ta^k$. If $w$ starts
with $\ta^{k-1}$, $w=\ta^{k-1}\tb^i\ta\tb^{k-i}$ follows  with $i\in[k-1]_{\geq 
2}$ and by 
Lemma~\ref{gensquares} the $k$-spectrum has $(i+1)k-4i+6\geq 3k-2>3k-4$ 
elements. 
By 
Lemma~\ref{lemhelp3} the claim follows for words starting with $(k-2)$ $\ta$. 
and it is shown that words starting with at least two and at most 
$k-3$ $\ta$ lead 
to 
$k$-spectra of cardinality greater than $3k-3$.\qed
\end{proof}

Going further, we analyse the larger possible cardinalities of $\ScatFact_k$, trying to see what values are achievable (even if only asymptotically, in some cases). 

\setcounter{corollary}{\value{definition}}
\begin{corollary}\label{corsquares}
All square numbers, greater or equal to four, occur as the cardinality of the $k$-spectrum of a word $w\in \Sigma^{2k}_{\ed}$; 
in particular 
$|\ScatFact_k(\ta^{\frac{k}{2}}\tb^k\ta^{\frac{k}{2}})|=\left(\frac{k}{2}+1
\right)^2$ holds for $k$ 
even.
\end{corollary}\refstepcounter{definition}

\begin{proof}
Apply Lemma~\ref{gensquares} to $i=\frac{k}{2}$. This implies that the 
cardinality of
the $k$-spectrum of $\ta^{\frac{k}{2}}\tb^k\ta^{\frac{k}{2}}$ is
\begin{align*}
k\left(\frac{k}{2}+1\right) - \frac{k^2}{4} -1 = 
\frac{1}{4}k^2+k-1=\left(\frac{k}{2}+1\right)^2.\qed
\end{align*}
\end{proof}

Inspired by the previous Corollary, we can show the following result concerning the asymptotic behaviour of the cardinality of $\ScatFact_k$ for words of length~$2k$.
\setcounter{proposition}{\value{definition}}
\begin{proposition}\label{nk}
Let $i>1$ be a fixed (constant) integer. Let $d= \lfloor\frac{k}{i}\rfloor$ and $r=k-di$, and 
$d'= \lfloor\frac{k}{i-1}\rfloor$ and $r'=k-d'(i-1)$ . Then the following hold:
\begin{itemize}
\item the word $\ta^r\tb^{r}(\ta^{d} \tb^d)^i$ has $\Theta(k^{2i-1})$ scattered 
factors of length $n$;
\item the word $\ta^r\tb^{r'}(\ta^{d} \tb^{d'})^{i-1}\ta^d$ has $\Theta(k^{2i-2})$ 
scattered factors of length $n$.
\end{itemize}
\end{proposition}\refstepcounter{definition}

\begin{proof}
Let us first show the upper bounds. The following algorithm can be used to find 
the scattered factors of length $k$ of  $\ta^r\tb^{r}(\ta^{d} \tb^d)^i$. Choose 
$2$ numbers $q_1$ and $q_2$ from $[i]_0$, and $2i-1$ integers 
$r_1,\ldots,r_{2i-1}$ from $[d]_0$. Let $r_{2i}=k- (q_1+q_2+\sum_{j\in 
[2i-1]}r_j)$. If $r_{2i}\geq 0$ then the word 
\[w'=a^{q_1}b^{q_2}(a^{r_1}b^{r_2})(a^{r_3}b^{r_4})\cdots 
(a^{r_{2i-1}}b^{r_{2i}})\] is a scattered factor of $\ta^r\tb^{r}(\ta^{d} 
\tb^d)^i$, and all scattered factors of length $k$ of this word have this form. 
From the construction of $w'$, because $d\leq \frac{k}{i}$, it follows that 
there are at most $O(i^2k^{2i-1})$ possible ways to obtain it. As $i$ is seen 
as 
a constant, this means that $\ta^r\tb^{r}(\ta^{d} \tb^d)^i$ has $O(n^{2i-1})$ 
scattered factors of length~$k$. 

In the same way one can show that  $\ta^r\tb^{r'}(\ta^{d} \tb^{d'})^{i-1}\ta^d$ 
has $O(n^{2i-2})$ scattered factors of length $n$.  

Let us now show the lower bounds. We first consider the word 
$\ta^r\tb^{r}(\ta^{d} \tb^d)^i$. As $i$ is constant, let us assume that 
$k>\frac{i(2i-1)}{i-1}$. Clearly, $\frac{k(i-1)}{i(2i-1)}< \frac{k}{2i-1} \leq 
\frac{k}{i}-1 \leq d\leq \frac{k}{i}$ and $d+r\geq \frac{k}{i}$. We generate 
scattered factors of the word $\ta^r\tb^{r}(\ta^{d} \tb^d)^i$ as follows. We 
firstly choose $2i-1$ integers $r_1,\ldots, r_{2i-1}$ between 
$\frac{k(i-1)}{i(2i-1)}$ and $\frac{k}{2i-1}$. Under our assumptions, the word 
\[w''=b^{r_1}(a^{r_2} b^{r_3})\cdots (a^{r_{2i-2}} b^{r_{2i-1}})\] is a 
scattered factor of the suffix $\tb^d(\ta^{d} \tb^d)^{i-1}$ of  
$\ta^r\tb^{r}(\ta^{d} \tb^d)^i$. Let $r_0=k-\sum_{j\in[2i-1]}r_j$. We have 
$r_0\leq \frac{k}{i}\leq d+r$, so $a^{r_0}w''$ is a scattered factor of 
$\ta^r(\ta^{d} \tb^d)^i$, so also of $\ta^r\tb^{r}(\ta^{d} \tb^d)^i$. Moreover, 
each choice of a tuple $(r_1,\ldots,r_{2i-1})$ leads to a different scattered 
factor of $\ta^r\tb^{r}(\ta^{d} \tb^d)^i$. The total number of tuples we choose 
is 
\[\left(\frac{k}{2i-1} 
-\frac{k(i-1)}{i(2i-1)}\right)^{2i-1}=\left(\frac{k}{i(2i-1)}\right)^{2i-1}.\] 

So the total number of scattered factors of length $k$ of $\ta^r\tb^{r}(\ta^{d} 
\tb^d)^i$ is at least $\left(\frac{k}{i(2i-1)}\right)^{2i-1}$. As the total 
number of scattered factors of length $k$ of $\ta^r\tb^{r}(\ta^{d} \tb^d)^i$ is 
also $O(k^{2i-1})$, we get that $\ta^r\tb^{r}(\ta^{d} \tb^d)^i$ has 
$\Theta(k^{2i-1})$ scattered factors of length $k$.

The proof that  $\ta^r\tb^{r'}(\ta^{d} \tb^{d'})^{i-1}\ta^d$ has 
$\Theta(n^{2i-2})$ scattered factors of length $k$ follows in a very similar 
manner.
\qed \end{proof}

\setcounter{remark}{\value{definition}}
\begin{remark}
Let $i$ be an integer, and consider $k$ another integer divisible by $i$. 
Consider the word $w_k=(a^{\frac{k}{i}} b^{\frac{k}{i}})^i$. The exact number of 
scattered factors of length $k$ of $w_k$ equals to the number 
$C\left(k,2i,\frac{k}{i}\right)$ of weak $2i$-compositions of $k$, whose terms 
are bounded by $\frac{k}{i}$, i.e., the number of ways in which $k$ can be 
written as a sum $\sum_{j\in[2i]}r_j$ where $r_j\in \left[\frac{k}{i}\right]_0$. From Proposition~\ref{nk} we also get that this number is $\Theta(n^{2k-1})$, but we also have: 
\[C\left(k,2i,\frac{k}{i}\right)=\sum_{0\leq j<M}(-1)^j{{2i}\choose{j}} 
{{k+2i-j(\frac{k}{i}+1)-1}\choose{2i-1}},\]
for $M=\frac{i(k+2i-1)}{k+i}.$
It is known that there exists a constant $E>0$ such that 
\[C\left(k,2i,\frac{k}{i}\right)\leq E\cdot \sum_{0\leq 
j<M}(-1)^j{{2i}\choose{j}} 
\left(k+2i-j\left(\frac{k}{i}+1\right)-1\right)^{2i-1}.\]
The coefficient of $k^{2i-1}$ in the right hand side of this inequality has to 
be positive. Consequently $\sum_{0\leq j<M}(-1)^j{{2i}\choose{j}} 
(i-j)^{2i-1}>0$.
 This seems to be an interesting combinatorial inequality in itself. 
 
One can also show as in Proposition~\ref{nk}  that the number of scattered factors of length $k$ of $w_k$, which have, 
at their turn, $(ab)^{i}$ as a scattered factor, is $\Theta(k^{2i-1})$. This 
number also equals the number $C'\left(k,2i,\frac{k}{i}\right)$ of $2i$-compositions 
of $k$ whose terms are strictly positive integers upper bounded by 
$\frac{k}{i}$, i.e., the number of ways in which $k$ can be written as a sum 
$\sum_{j\in[2i]}r_j$ where $r_j\in \left[\frac{k}{i}\right]$. Just as above, 
from this we get $\sum_{0\leq j<i}(-1)^j{{2i}\choose{j}} 
(i-j)^{2i-1}>0$.
Again, this inequality seems interesting to us.
\end{remark}\refstepcounter{definition}

We will end this analysis with the conjecture that, in contrast to the first gap, 
which always starts 
immediately after the first obtainable cardinality, the last gap ends earlier 
the larger $k$ is. 
More precisely, if $w=\ta^2\tb^2(\ta\tb)^{k-3-i}\tb\ta(\ta\tb)^i$ for 
$k\in\N_{\geq 4}$, $i\in[k-2]_0$  then $|\ScatFact_k(w)|=2^k-2-i$.

At the end of this section, we will briefly introduce $\theta$-palindromes in 
this specific 
setting. Let $\theta:\Sigma^{\ast}\rightarrow\Sigma^{\ast}$ be an antimorphic 
involution, i.e. 
$\theta(uv)=\theta(v)\theta(u)$ and $\theta^2$ is the identity on 
$\Sigma^{\ast}$. By $\Sigma=\{\ta,\tb\}$ only the identity and renaming are 
such mappings.  The fixed points of $\theta$ are called $\theta$-palindromes 
($\ta\tb^3.\theta(\tb)^3\theta(\ta)$) and exactly the words where 
$w^R=\overline{w}$ holds. They 
were studied in different fields well (see e.g., 
\cite{journals/iandc/FazekasMMS14}, 
\cite{journals/nc/KariM10}). A word $w\in\Sigma_{\ed}^{2k}$ is a 
$\theta$-palindrome 
iff either $w\in\{\ta w'\tb,\tb w' \ta\}$ for some $\theta$-palindrome 
$w'\in\Sigma_{\ed}^{2(k-1)}$ 
or additionally $w=\ta^{\frac{k}{2}}\tb^k\ta^{\frac{k}{2}}$ in the case that $k$ 
is even. Two 
cardinality results for $\theta$-palindromes are presented in 
Lemma~\ref{gensquares} and 
Corollary~\ref{corsquares}. We believe that persuing the $k$-spectra of 
$\theta$-palindromes may 
lead to a deeper insight of which cardinalities can be reached, but due to space 
restrictions we 
will only mention one conjecture here, which may already show that cardinalities 
are somehow 
propagating for $\theta$-palindromes. Notice that this conjecture implies that 
indeed similar to the second gap here 
$4k-4$ is always 
reached but that in contrast to the second gap, the third gap is not of the 
form 
$4k-4-i$ for 
$i\in[k-4]$.

\setcounter{conjecture}{\value{definition}}
\begin{conjecture}
The $k$-spectrum of $w=\ta\tb^{k-1}\ta^{k-1}\tb$ has $4(k-1)$ elements and 
moreover if $w'=w^R$ with a $k$-spectrum of cardinality $\ell\in\N_{\geq 12}$
then the scattered factor set of $\ta w\tb$ has cardinality $2\frac{1}{4}\ell 
-5$.
\end{conjecture}\refstepcounter{definition}

\section{Reconstructing Weakly-$0$-Balanced Words from their $k$-Spectra}
\label{reconstruct}
In the final section we consider the slightly different problem of reconstructing a word from its scattered factors, or more specifically in this case, $k$-spectra. More generally, we are interested in how much information about a (weakly-$0$-balanced) word $w$ is 
contained in its scattered factors, and more precisely, which scattered factors 
are not necessary or useful for reconstructing the word $w$, or distinguishing it from others. Since $w$ is a scattered factor of itself, it is trivial that the 
scattered factor of length $|w|$ is sufficient to uniquely reconstruct $w$. On 
the other hand, all words over $\{\ta,\tb\}^*$ containing both letters will have the 
same $1$-spectrum. Thus we see that the length of the scattered factors of a 
word $w$ plays a role in how much information about $w$ they contain. This 
relationship is described more precisely by the following result of Dress and 
Erd\"{o}s~\cite{dress04} along with the fact that (cf. e.g. 
Proposition~\ref{lem:full}) a word of length $2k$ is not uniquely determined by 
its scattered factors of length $k$.

\setcounter{proposition}{\value{definition}}
\begin{proposition}[Dress and Erd\"{o}s \cite{dress04}]\label{prop:erdos}
If $\ScatFact_{k+1}(w)= \ScatFact_{k+1}(w^\prime)$ holds for $w, w^\prime \in 
\Sigma^{\leq 2k}$ 
then $w = w^\prime$ follows.
\end{proposition}\refstepcounter{definition}

\begin{proof}[for $w$, $w'$ being
weakly-$0$-balanced]
We give a procedure for uniquely reconstructing $w$ from $\ScatFact_k(w)$. For 
all 
$i,j \in\mathbb{N}_0$ such that $i+j = k$, ask whether $\ta^i\tb\ta^j \in 
\ScatFact_k(w)$. Since 
there are exactly $i+j$ occurrences of $\ta$ in $w$, all are accounted for in 
the (potential) 
scattered factor $\ta^i\tb\ta^j$, and thus the answer is `yes' if and only if 
there are one or more 
$\tb$s between the \nth{$i$} and \nth{$(i+1)$} occurrences of $\ta$ in $w$. Hence after 
these queries, we 
know 
exactly which $\ta$s are consecutive (i.e. do not have a $\tb$ between them) in 
$w$. Similarly we 
ask for all $i,j \in\mathbb{N}_0$ such that $i+j = k$, ask whether 
$\tb^i\ta\tb^j \in 
\ScatFact_k(w)$. By symmetry, this tells us exactly which $\tb$s are 
consecutive. This is 
sufficient 
information to specify $w$ completely.\qed
\end{proof}

In the proof of Proposition~\ref{prop:erdos}, 
a pivotal role is played by scattered factors which 
contain many $\ta$s and a few $\tb$s or vice-versa. The 
question arises as to whether this is due to the fact that these 
scattered factors contain inherently more information about the structure of the 
whole word than e.g., weakly-$0$-balanced ones. In the general case, the answer is, 
sometimes at least, yes: we cannot distinguish between e.g. two words in 
$\{\ta\}^*$ by their weakly-$0$-balanced scattered factors, as the only such factor is $\varepsilon$. The same problem arises for all words which have a sufficiently uneven ratio of $\ta$s to $\tb$s. 

However, if in addition we consider only weakly-$0$-balanced words, then the 
situation changes. We conjecture that in fact, for these words $w$, 
the weakly-$0$-balanced scattered factors are just as informative about the $w$ as 
the unbalanced ones. More formally, we believe the following adaptation of 
Proposition~\ref{prop:erdos} holds:

\setcounter{conjecture}{\value{definition}}
\begin{conjecture}\label{con:bounds}
Let $k \in \mathbb{N}$. Let $k^\prime = k+1$ for odd $k$, and $k^\prime = k+2$ 
for even~$k$.~Let 
$w, w^\prime \in \Sigma_{\ed}^{2k}$ such that $\ScatFact_{k^\prime}(w) \cap 
\Sigma^{k'}_{\ed}= 
\ScatFact_{k^\prime}(w^\prime) \cap \Sigma^{k'}_{\ed}$. Then $w = w^\prime$.
\end{conjecture}
\refstepcounter{definition}

While we do not resolve the conjecture, we give an example of a subclass of 
words for which it holds true, namely when there are at most two blocks of 
$\tb$s (and therefore by symmetry if there are at most two blocks of 
$\ta$s).

\setcounter{proposition}{\value{definition}}
\begin{proposition}\label{propjoel}
Let $k\in \mathbb{N}$. If $k$ is odd, then each word $w \in 
\ta^*\tb^*\ta^*\tb^*\ta^* \cap 
\Sigma_{\ed}^{2k}$ is uniquely determined by the set $\ScatFact_{k+1}(w) \cap 
\Sigma_{\ed}^{k+1}$. 
Similarly, if $k$ is even, then each word $w \in \ta^*\tb^*\ta^*\tb^*\ta^* \cap 
\Sigma_{\ed}^{2k}$ 
is uniquely determined by the set $\ScatFact_{k+2}(w) \cap \Sigma_{\ed}^{k+2}$.
\end{proposition}\refstepcounter{definition}

\begin{proof}
As in the proof of Proposition~\ref{prop:erdos}, we give an algorithm for 
uniquely reconstructing $w$. 
W.l.o.g., let $k$ be odd. The case that $k$ is even is easily adapted. Let $w = 
\ta^i\tb^j\ta^\ell\tb^{k-j}\ta^{k-i-\ell}$ and let $S = \ScatFact_{k+1}(w) \cap 
\Sigma^*_{\ed}$.
Firstly, we shall deal with the case that $\ell=0$. Note that we can decide 
whether $\ell =0$ by querying whether there exists a scattered factor $u \in S$ 
such that $u \in \ta^*\tb^+\ta^+\tb^+\ta^*$. Now, if $\ell =0$, we have $w = 
\ta^i \tb^k\ta^{k-i}$. Since $k$ is odd, exactly one of $i,k-i$ will be at most 
$\frac{k-1}{2}$. We can decide which one by querying whether 
$\ta^{\frac{k+1}{2}}\tb^{\frac{k+1}{2}} \in S$. W.l.og., suppose $i \leq 
{\frac{k-1}{2}}$ (so the query returns ``no''). The other case is symmetric. 
Then note that $\ta^i\tb^{\frac{k+1}{2}}\ta^{\frac{k+1}{2}-i} \in S$ but 
$\ta^{i+1}\tb^{\frac{k+1}{2}}\ta^{\frac{k+1}{2}-i-1} \notin S$. Thus the exact 
value of $i$ (and therefore $k-i$) can be inferred directly from observing 
scattered factors of this form in $S$.

Now consider the the case that $\ell \not= 0$. Note that there exists $u \in 
\tb^+ \ta^{\frac{k+1}{2}} \tb^+ \cap S$ if and only if $\ell \geq 
\frac{k+1}{2}$. Suppose firstly that $\ell \geq \frac{k+1}{2}$. Then $i + 
(k-i-\ell) \leq \frac{k-1}{2}$. Thus we can determine $i$ and $(k-i-\ell)$ (and 
therefore $\ell$) by looking for the maximum $m_1,m_2$ such that there exists $u 
\in \ta^{m_1}\tb^+\ta^+\tb^+\ta^{m_2}$ with $u \in S$ ($i$ is the maximum $m_1$ 
while $k-i-\ell$ is the maximum $m_2$). Moreover, exactly one of $j,k-j$ will be 
less than $\frac{k+1}{2}$. We can decide which one by querying whether 
$\ta^{\frac{k+1}{2}} \tb^{\frac{k+1}{2}} \in S$. If so, it must be that $k-j 
\geq {\frac{k+1}{2}}$. Suppose that this is the case (the other case is 
symmetric). Then as before, we can determine the exact value of $j$ by looking 
at the scattered factors of the form $\tb^m\ta^+\tb^+\ta^*$ (i.e., $j$ is the 
maximum $m$) and we are done.

Finally, we consider the case that $0 <\ell <{\frac{k+1}{2}}$. Then $\ell$ can 
be uniquely determined as the maximum $m$ such that there exists $u \in 
\ta^*\tb^+\ta^m\tb^+\ta^*$ with $u \in S$. In order to determine $i$ (or 
equivalently $k-i-\ell$), we look for the maximum $m_1,m_2$ such that there 
exist $u_1\in \ta^{m_1} \tb^+\ta^+\tb^+\ta^*$ and $u_2 \in \ta^* 
\tb^+\ta^+\tb^+\ta^{m_2}$ with $u_1,u_2\in S$. In particular at least one of 
$m_1,m_2$ must be strictly less than $\frac{k-1}{2}$. If $m_1 < \frac{k-1}{2}$, 
then $j=m_1$ and if $m_2 < {\frac{k-1}{2}}$ then $k-\ell-i = m_2$. In either 
case, since $\ell$ is already known, this uniquely determines both $i$ and 
$k-i-\ell$. 

It remains to determine $j$ (or equivalently $k-j$). Recall that exactly one of 
$j,k-j$ will be less than $\frac{k+1}{2}$. Let $m_1$ be the maximum $m$ such 
that there exists $u\in \ta^*\tb^m\ta^+\tb^+\ta^*$ with $u \in S$ and let $m_2$ 
be the maximum $m$ such that there exists $u\in\ta^*\tb^+\ta^+\tb^m\ta^*$ with 
$u\in S$. Note that $m_1,m_2\leq \frac{k-1}{2}$. If $m_1 < \frac{k-1}{2}$ 
(resp. 
$m_2 < \frac{k-1}{2}$), then $j= m_1$ (resp $k-j = m_2$), and thus $j$ and 
$j-k$ 
can be inferred. If $m_1 = m_2 = \frac{k-1}{2}$, then either $j = 
\frac{k-1}{2}$ 
or $k-j = \frac{k-1}{2}$. Now, if $k-i-\ell < \frac{k+1}{2}$, there exists $u 
\in \ta^*\tb^{\frac{k+1}{2}}\ta^+\ta^{k-i-\ell}$ with $u\in S$ if and only if 
$j= \frac{k+1}{2}$ (in which case $k-j = \frac{k-1}{2}$). On the other hand, if 
$k-i-\ell \geq \frac{k+1}{2}$, then $i<\frac{k+1}{2}$ and there exists $u\in 
\ta^i\ta^+\tb^{\frac{k+1}{2}}\ta^*$ with $u\in S$ if and only if $k-j = 
\frac{k+1}{2}$ (in which case $j= \frac{k-1}{2}$). In either case, all 
exponents 
are known and we have uniquely reconstructed $w$.\qed
\end{proof}

The difficulty in proving Conjecture~\ref{con:bounds} seems to arise from the 
fact that, for different pairs of words $w,w’ \in \Sigma_{\ed}$, the set of 
scattered factors which distinguish them, namely the symmetric difference of 
$\ScatFact_k(w) \cap \Sigma^{k}_{\ed}$  and $\ScatFact_k(w') \cap \Sigma^{k}_{\ed}$ 
(for appropriate $k$), varies considerably, unlike with the proof(s) of 
Proposition~\ref{prop:erdos}, where the set of distinguishing scattered factors 
is always made up words of the same form, regardless of the choice of $w$ and 
$w’$. As an example, consider 
the words $w = \ta\tb\ta\tb\ta\tb$, $w^\prime = \tb\ta\tb\ta\tb\ta$, and 
$w^{\prime\prime} = 
\ta\tb\ta\tb\tb\ta$. Then the symmetric difference of $\ScatFact_{4}(w) \cap 
\Sigma_{\ed}^{4}$ and 
$\ScatFact_{4}(w^\prime) \cap \Sigma_{\ed}^{4}$ is $\{\ta\ta\tb\tb, 
\tb\tb\ta\ta\}$. On the other 
hand, considering $\ScatFact_{4}(w^{\prime}) \cap \Sigma_{\ed}^{4}$ and 
$\ScatFact_{4}(w^{\prime\prime}) \cap \Sigma_{\ed}^{4}$, the symmetric 
difference is 
$\{\tb\ta\ta\tb\}$.

\section{Conclusions}
We have considered properties of $k$-spectra of weakly-$0$-balanced words. In 
particular, in Section~3 
we give several insights into the structure of the set of all $k$-spectra of 
weakly-$0$-balanced words 
of length $2k$ by considering for which numbers $n$ there exists $w$ such that 
the $k$-spectrum of 
$w$ has cardinality $n$. In particular, we characterise the first two gaps in 
the possibilities for 
each $k$ which are regular (in the sense that the first and second gaps are 
always from $k+2$ to 
$2k-1$ and $2k+1$ to $3k-4$ (inclusive). On the other hand, we see that the 
third gap is 
considerably 
less regular and thus resists a natural characterisation.

In Section~4, we consider the task of reconstructing 
weakly-$0$-balanced words 
from their $k$-spectra. We note that this is, in a sense, as hard as in the 
general case, however, we also conjecture that even if we consider only the 
scattered factors which are also weakly-$0$-balanced, then the situation remains 
the same, in the sense that it can be achieved for the same choices of $k$. 
Resolving this conjecture appears to require some new approach however since the 
techniques for the general case are not easily adapted.

As mentioned at the end of Section~3 some of the weakly-$0$-balanced words are 
$\theta$-palindromes. 
Since the $\theta$-palindromes of length $2k$ are constructible from the ones of 
length $2(k-1)$ 
(except for each even $k$ exactly one $\theta$-palindrome) we surmised that the 
structure and 
properties propagate. Moreover we expected that the knowledge of the word's 
second half helps in 
finding the cardinalities of the $k$-spectra. Nevertheless we were only able to 
get results for 
$\theta$-palindromes in the same manner as for the other words, but we still 
believe that the 
structure of the $\theta$-palindromes can reveal more insights with further 
work.

\newpage

\bibliographystyle{plain}
\bibliography{scatfact}

%
%
%
%
%
%
%
%

\end{document}